\documentclass[11pt]{article}

\usepackage{amsmath,amsthm,amsfonts,amssymb}
\usepackage{ifthen}

\usepackage{bbm}

\usepackage{tikz}
\usetikzlibrary{arrows.meta,positioning,fit,calc,intersections,patterns}
\usepackage{pgfplots}
\usepgfplotslibrary{fillbetween}
\usepackage{soul}
\usetikzlibrary{arrows}
\usepackage{nicefrac}
\usepackage{enumitem}

\usepackage{multirow}

\usepackage{marginfix}

\usepackage[round]{natbib}
\usepackage[capitalize,nameinlink,noabbrev]{cleveref}

\newtheorem{theorem}{Theorem}
\newtheorem{corollary}{Corollary}
\newtheorem{lemma}{Lemma}
\newtheorem{definition}{Definition}
\newtheorem{fact}{Fact}
\newtheorem{proposition}{Proposition}

\newtheorem{example}{Example}

\usepackage{chngcntr}
\counterwithin{figure}{section}
\counterwithin{theorem}{section}
\counterwithin{corollary}{section}
\counterwithin{lemma}{section}
\counterwithin{proposition}{section}
\counterwithin{fact}{section}
\counterwithin{example}{section}
\counterwithin{definition}{section}

\usepackage{caption}
\usepackage{subcaption}

\newcommand{\addtag}{\refstepcounter{equation}\tag{\theequation}}

\newcommand{\prob}[2][]{\text{\bf Pr}\ifthenelse{\not\equal{}{#1}}{_{#1}}{}\!\left[{\def\givenn{\middle|}#2}\right]}
\newcommand{\expect}[2][]{\mathbb{E} \ifthenelse{\not\equal{}{#1}}{_{#1}}{}\!\left[{\def\givenn{\middle|}#2}\right]}

\DeclareMathOperator{\argmax}{argmax}

\newcommand{\given}{\,\mid\,}

\newcommand{\signal}{x}

\newcommand{\distof}[1]{\Delta_{#1}}

\newcommand{\val}{v}

\DeclareMathOperator{\OPT}{OPT}
\DeclareMathOperator{\opt}{opt}
\DeclareMathOperator{\ALG}{ALG}
\DeclareMathOperator{\FTL}{FTL}

\DeclareMathOperator{\EW}{EW}
\DeclareMathOperator{\Regret}{Regret}

\DeclareMathOperator{\MAB}{MAB}
  \DeclareMathOperator{\SDA}{SDA}
  \DeclareMathOperator{\OLA}{OLA}
 \DeclareMathOperator{\btl}{btl}
 \DeclareMathOperator{\BTL}{BTL}
 \DeclareMathOperator{\BTPL}{PBTL}
  \DeclareMathOperator{\MAXPTRB}{MAXPTRB}
 \DeclareMathOperator{\FTPL}{PFTL}

 \DeclareMathOperator{\OSV}{OSV}
  \DeclareMathOperator{\SV}{SV}

\newcommand{\olv}{u}

\newcommand{\olvs}{{\mathbf \olv}}
\newcommand{\eolvs}{\tilde{\olvs}}
\newcommand{\eolv}{\tilde{\olv}}

\newcommand{\vali}[1][i]{v_{#1}}
\newcommand{\vals}{{\mathbf v}}

\newcommand{\alloc}{x}
\newcommand{\ballocs}{{\mathbf \alloc}}
\newcommand{\pay}{p}
\newcommand{\bpays}{{\mathbf \pay}}
\newcommand{\bids}{{\mathbf b}}
\newcommand{\bidsmi}[1][i]{{\mathbf b}_{-#1}}

\newcommand{\balloci}[1][i]{\alloc_{#1}}
\newcommand{\bpayi}[1][i]{\pay_{#1}}

\newcommand{\cumval}{U}

\newcommand{\regret}{r}
\newcommand{\regreti}[1][i]{\regret_{#1}}
\newcommand{\eregret}{\tilde{\regret}}
\newcommand{\eregreti}[1][i]{\eregret_{#1}}
\newcommand{\evali}[1][i]{\tilde{\val}_{#1}}

\newcommand{\olact}{a}
\newcommand{\pr}{\alpha}
\newcommand{\prs}{{\boldsymbol \pr}}
\newcommand{\epr}{\tilde{\pr}}
\newcommand{\eprs}{\tilde{\prs}}
\newcommand{\Qmat}{\tilde{A}}

\newcommand{\pvec}{\prs}
\newcommand{\qvec}{\eprs}

\newcommand{\pij}{\pr^i_\olact}
\newcommand{\pveci}[1][i]{\pvec^{#1}}

\newcommand{\vveci}[1][i]{\mathbf{\olvs}^{#1}}

\newcommand{\zee}{z}

\newcommand{\roundalt}{{i'}}

\newcommand{\rowprob}{\alpha}
\newcommand{\rowstrat}{\boldsymbol{\rowprob}}

\newcommand{\colprob}{\beta}
\newcommand{\colstrat}{\boldsymbol{\colprob}}

\newcommand{\regretcond}{\rho}
\newcommand{\regretvec}{\boldsymbol{\regretcond}}
 
\usepackage{fullpage}

\def\xMin{0}
\def\xMax{1}

\newcommand{\mabreductionfig}{\begin{tikzpicture}[>=stealth, node distance=4.2cm, font=\small]

\node[draw, rounded corners, minimum width=1.5cm, minimum height=2.8cm, align=center] (ola)    {\textbf{OLA}};
\node[draw, rounded corners, minimum width=1.5cm, minimum height=2.8cm, align=center, right=of ola] (mab) {\textbf{MAB}};
\node[draw, rounded corners, minimum width=1.5cm, minimum height=2.8cm, align=center, right=of mab] (world) {\textbf{WORLD}};

\path let \p1 = (ola.south),   \p2 = (ola.north)   in coordinate (olaR)   at ($(\p1)!0.6666!(\p2)$);
\path let \p1 = (mab.south),   \p2 = (mab.north)   in coordinate (mabR)   at ($(\p1)!0.6666!(\p2)$);
\path let \p1 = (world.south), \p2 = (world.north) in coordinate (worldR) at ($(\p1)!0.6666!(\p2)$);

\path let \p1 = (ola.south),   \p2 = (ola.north)   in coordinate (olaL)   at ($(\p1)!0.3333!(\p2)$);
\path let \p1 = (mab.south),   \p2 = (mab.north)   in coordinate (mabL)   at ($(\p1)!0.3333!(\p2)$);
\path let \p1 = (world.south), \p2 = (world.north) in coordinate (worldL) at ($(\p1)!0.3333!(\p2)$);

\draw[->, thick] (olaR -| ola.east) -- node[above]{\(\eprs\)} (mabR -| mab.west);
\draw[->, thick] (mabR -| mab.east) -- node[above]{\(\olact \sim \prs=(1-\varepsilon)\eprs+\varepsilon/k\)} (worldR -| world.west);

\draw[->, thick] (mabL -| mab.west) -- node[below]{\(\eolvs=\bigl(0,\ldots,\tfrac{\olv_\olact}{\pr_\olact},\ldots,0\bigr)\)} (olaL -| ola.east);
\draw[->, thick] (worldL -| world.west) -- node[below]{\(\olv_\olact\)} (mabL -| mab.east);

\end{tikzpicture}}

\newcommand{\swapreductionfig}{\begin{tikzpicture}[>=stealth, font=\small]

\def\boxw{2.1cm}
\def\boxh{0.8cm}
\def\vsep{0.4cm}
\def\hgap{3.5cm} 

\node[draw, rounded corners, minimum width=\boxw, minimum height=\boxh, align=center] (A1) {$\OLA_{1}$};
\node[draw, rounded corners, minimum width=\boxw, minimum height=\boxh, align=center,
      below=\vsep of A1] (A2) {$\OLA_{2}$};
\node[align=center, below=\vsep of A2] (Adots) {$\vdots$};
\node[draw, rounded corners, minimum width=\boxw, minimum height=\boxh, align=center,
      below=\vsep of Adots] (Ak) {$\OLA_{k}$};

\coordinate (TopRef) at ($(A1.north east)+(\hgap,0)$);
\coordinate (BotRef) at ($(Ak.south east)+(\hgap,0)$);

\node[draw, rounded corners, fit=(TopRef)(BotRef),
      inner xsep=0.8cm, inner ysep=0pt] (H) {};
\node[align=center] at (H.center) {$\SDA$};

\coordinate (WTop) at ($(TopRef)+(\hgap+0.8cm,0)$);
\coordinate (WBot) at ($(BotRef)+(\hgap+0.8cm,0)$);
\node[draw, rounded corners, fit=(WTop)(WBot),
      inner xsep=0.9cm, inner ysep=0pt] (World) {};
\node[align=center] at (World.center) {WORLD};

\foreach \name in {A1,A2,Ak}{
  \coordinate (\name-66) at ($(\name.south east)!0.6666!(\name.north east)$);
  \coordinate (\name-33) at ($(\name.south east)!0.3333!(\name.north east)$);
}

\foreach \name in {A1,A2,Ak}{
  \coordinate (H-\name-66) at ($(H.west |- \name-66)$);
  \coordinate (H-\name-33) at ($(H.west |- \name-33)$);
}

\draw[->, thick] (A1-66) -- node[above]{$\qvec_1$} (H-A1-66);
\draw[->, thick] (A2-66) -- node[above]{$\qvec_2$} (H-A2-66);
\draw[->, thick] (Ak-66) -- node[above]{$\qvec_k$} (H-Ak-66);

\draw[->, thick] (H-A1-33) -- node[below]{$\eolvs_1 = \pr_1\,\olvs$} (A1-33);
\draw[->, thick] (H-A2-33) -- node[below]{$\eolvs_2 =\pr_2\,\olvs$} (A2-33);
\draw[->, thick] (H-Ak-33) -- node[below]{$\eolvs_k =\pr_k\,\olvs$} (Ak-33);

\coordinate (H-e66) at ($(H.south east)!0.6666!(H.north east)$);
\coordinate (H-e33) at ($(H.south east)!0.3333!(H.north east)$);
\coordinate (W-w66) at ($(World.south west)!0.6666!(World.north west)$);
\coordinate (W-w33) at ($(World.south west)!0.3333!(World.north west)$);

\draw[->, thick] (H-e66) -- node[above]{$\olact \sim \pvec$} (W-w66);

\draw[->, thick] (W-w33) -- node[below]{$\olvs$} (H-e33);

\end{tikzpicture}}

\newcommand{\btlvsoptfig}{\begin{tikzpicture}[x=3cm,y=1cm,font=\small,rounded corners=1pt]

\def\barh{0.6}
\def\gap{0.30}

\fill[black!20,rounded corners=1pt] (0.0,0) rectangle (0.4,\barh);   \fill[black!20,rounded corners=1pt] (0.4,0) rectangle (0.7,\barh);   \fill[black!20,rounded corners=1pt] (0.7,-\barh-\gap) rectangle (1.4,-\gap); \fill[black!20,rounded corners=1pt] (1.4,0) rectangle (2.0,\barh);   

\fill[pattern=north east lines,pattern color=black,rounded corners=1pt] (0.0,0) rectangle (0.4,\barh);
\fill[pattern=north east lines,pattern color=black,rounded corners=1pt] (0.4,0) rectangle (0.7,\barh);
\fill[pattern=north east lines,pattern color=black,rounded corners=1pt] (0.4,-\barh-\gap) rectangle (1.4,-\gap);
\fill[pattern=north east lines,pattern color=black,rounded corners=1pt] (1.0,0) rectangle (2.0,\barh);

\foreach \x/\xx in {0/0.4,0.4/0.7,0.7/1.0,1.0/2.0}
  \draw[line width=0.5pt,rounded corners=1pt] (\x,0) rectangle (\xx,\barh);
\foreach \x/\xx in {0/0.2,0.2/0.4,0.4/1.4,1.4/1.6}
  \draw[line width=0.5pt,rounded corners=1pt] (\x,-\barh-\gap) rectangle (\xx,-\gap);

\foreach \x/\label in {0.2/1,0.55/2,0.85/3,1.5/4}
  \node[fill=white,inner sep=1pt,rounded corners=1pt] at (\x,0.5*\barh) {\label};
\foreach \x/\label in {0.1/1,0.3/2,0.9/3,1.5/4}
  \node[fill=white,inner sep=1pt,rounded corners=1pt] at (\x,-\gap-0.5*\barh) {\label};

\node[anchor=east] at (-0.25, 0.5*\barh) {Action 1};
\node[anchor=east] at (-0.25,-\gap-0.5*\barh) {Action 2};

\begin{scope}[shift={(-.3,-2)}]
\def\lw{0.35}   \def\lh{\barh}  

\fill[pattern=north east lines,pattern color=black,rounded corners=1pt]
    (0,0) rectangle +( \lw,\lh );
  \draw[rounded corners=1pt,line width=0.5pt]
    (0,0) rectangle +( \lw,\lh );
\node[fill=white,rounded corners=1pt,inner sep=1pt]
    at (0+0.5*\lw, 0+0.5*\lh) {$i$};
\node[anchor=west] at (\lw+0.06, 0.5*\lh) {$\btl^i$};

\begin{scope}[shift={(1,0)}]
    \fill[black!20,rounded corners=1pt] (0,0) rectangle +( \lw,\lh );
    \draw[rounded corners=1pt,line width=0.5pt] (0,0) rectangle +( \lw,\lh );
    \node[fill=white,rounded corners=1pt,inner sep=1pt]
      at (0+0.5*\lw, 0+0.5*\lh) {$i$};
    \node[anchor=west] at (\lw+0.06, 0.5*\lh) {$\opt^i$};
  \end{scope}
\end{scope}

\end{tikzpicture}}

\newcommand{\pftlX}{.7cm}

\newcommand{\pftlBarH}{0.6}
\newcommand{\pftlRC}{1pt}
\newcommand{\pftlCoinXShift}{0.55}
\newcommand{\pftlLblInset}{0.10}
\newcommand{\pftlTotalX}{7.4} \newcommand{\pftlLeftXmin}{-0.8} 

\newcommand{\pftlYAone}{0}
\newcommand{\pftlYAtwo}{-1.0}
\newcommand{\pftlYAthree}{-2.0}
\newcommand{\pftlYAfour}{-3.0}

\newcommand{\pftlWAone}{3.2}
\newcommand{\pftlWAtwo}{2.4}
\newcommand{\pftlWAthree}{5.6}
\newcommand{\pftlWAfour}{4.8}

\newcommand{\pftlDrawBar}[4][Action]{\coordinate (btl@base) at (0,#2);
  \fill[black!20,rounded corners=\pftlRC] (btl@base) rectangle ++(#3,\pftlBarH);
  \draw[line width=0.5pt,rounded corners=\pftlRC] (btl@base) rectangle ++(#3,\pftlBarH);
\node[anchor=east] at (-0.40,#2+0.5*\pftlBarH) {#1 #4};
\node[anchor=east,inner sep=1pt] at (#3-\pftlLblInset,#2+0.5*\pftlBarH) {$\cumval^{i-1}_{#4}$};
}

\newcommand{\pftlDrawCoin}[2]{\draw[line width=0.5pt] (#2+\pftlCoinXShift,#1+0.5*\pftlBarH) circle (0.18);
  \draw[line width=0.3pt] (#2+\pftlCoinXShift,#1+0.5*\pftlBarH) circle (0.13);
  \node[scale=0.6] at (#2+\pftlCoinXShift,#1+0.5*\pftlBarH) {\$};
}

\newcommand{\pftlDrawBrokenCoin}[3]{\draw[line width=0.5pt,dashed,color=black!50] (#2+\pftlCoinXShift,#1+0.5*\pftlBarH) circle (0.18);
  \draw[line width=0.3pt,dashed,color=black!50] (#2+\pftlCoinXShift,#1+0.5*\pftlBarH) circle (0.13);
  \node[scale=0.8,text=black!50] at (#2+\pftlCoinXShift,#1+0.5*\pftlBarH) {#3};
}

\newcommand{\pftlAppendBox}[4]{\draw[line width=0.5pt,rounded corners=\pftlRC] (#2,#1) rectangle ++(#3,\pftlBarH);
\coordinate (btl@bl) at (#2,#1);
  \coordinate (btl@tr) at ($(#2,#1)+(#3,\pftlBarH)$);
  \coordinate (btl@C)  at ($ (btl@bl)!0.5!(btl@tr) $);
\draw[line width=0.5pt,color=black!50] (btl@C) circle (0.18);
  \draw[line width=0.3pt,color=black!50] (btl@C) circle (0.13);
  \node[scale=0.8,text=black!50] at (btl@C) {#4};
}

\newcommand{\pftlAppendGrayBox}[4]{\fill[black!20,rounded corners=\pftlRC] (#2,#1) rectangle ++(#3,\pftlBarH);
  \draw[line width=0.5pt,rounded corners=\pftlRC] (#2,#1) rectangle ++(#3,\pftlBarH);
  \node at (#2+0.5*#3,#1+0.5*\pftlBarH) {#4};
}

\newcommand{\pftlreadycoinsfig}{\begin{tikzpicture}[x=\pftlX,y=1cm,font=\small,rounded corners=\pftlRC]
\path[use as bounding box] (\pftlLeftXmin,-3.5) rectangle (\pftlTotalX,1);

  \pftlDrawBar[]{\pftlYAone}{\pftlWAone}{1}
  \pftlDrawBar[]{\pftlYAtwo}{\pftlWAtwo}{2}
  \pftlDrawBar[]{\pftlYAthree}{\pftlWAthree}{3}
  \pftlDrawBar[]{\pftlYAfour}{\pftlWAfour}{4}

  \pftlDrawCoin{\pftlYAone}{\pftlWAone}
  \pftlDrawCoin{\pftlYAtwo}{\pftlWAtwo}
  \pftlDrawCoin{\pftlYAthree}{\pftlWAthree}
  \pftlDrawCoin{\pftlYAfour}{\pftlWAfour}
\end{tikzpicture}}

\newcommand{\pftlflippedcoinsfig}{\begin{tikzpicture}[x=\pftlX,y=1cm,font=\small,rounded corners=\pftlRC]
  \path[use as bounding box] (\pftlLeftXmin,-3.5) rectangle (\pftlTotalX,1);
  
  \pftlDrawBar[]{\pftlYAone}{\pftlWAone}{1}
  \pftlDrawBar[]{\pftlYAtwo}{\pftlWAtwo}{2}
  \pftlDrawBar[]{\pftlYAthree}{\pftlWAthree}{3}
  \pftlDrawBar[]{\pftlYAfour}{\pftlWAfour}{4}

\pftlDrawBrokenCoin{\pftlYAone}{\pftlWAone}{2}

\pftlAppendBox{\pftlYAtwo}{\pftlWAtwo}{1}{1}
  \pftlAppendBox{\pftlYAtwo}{\pftlWAtwo+1}{1}{3}
  \pftlDrawBrokenCoin{\pftlYAtwo}{\pftlWAtwo+2}{4}

\pftlDrawBrokenCoin{\pftlYAthree}{\pftlWAthree}{6}

\pftlAppendBox{\pftlYAfour}{\pftlWAfour}{1}{5}
  \pftlDrawCoin{\pftlYAfour}{\pftlWAfour+1}

\end{tikzpicture}}

\newcommand{\pftlfinalcoinfig}{\begin{tikzpicture}[x=\pftlX,y=1cm,font=\small,rounded corners=\pftlRC]
  \path[use as bounding box] (\pftlLeftXmin,-3.5) rectangle (\pftlTotalX,1);

  \pftlDrawBar[]{\pftlYAone}{\pftlWAone}{1}
  \pftlDrawBar[]{\pftlYAtwo}{\pftlWAtwo}{2}
  \pftlDrawBar[]{\pftlYAthree}{\pftlWAthree}{3}
  \pftlDrawBar[]{\pftlYAfour}{\pftlWAfour}{4}

\pftlAppendGrayBox{\pftlYAone}{\pftlWAone}{1}{$\olv^{i}_{1}$}

\pftlAppendBox{\pftlYAtwo}{\pftlWAtwo}{1}{1}
  \pftlAppendBox{\pftlYAtwo}{\pftlWAtwo+1}{1}{3}
  \pftlAppendGrayBox{\pftlYAtwo}{\pftlWAtwo+2}{1}{$\olv^{i}_{2}$}

\pftlAppendGrayBox{\pftlYAthree}{\pftlWAthree}{1}{$\olv^{i}_{3}$}

\pftlAppendBox{\pftlYAfour}{\pftlWAfour}{1}{5}
  \pftlAppendBox{\pftlYAfour}{\pftlWAfour+1}{1}{7}
  \pftlAppendGrayBox{\pftlYAfour}{\pftlWAfour+2}{0.6}{$\olv^{i}_{4}$}
\end{tikzpicture}}

\newcommand{\inferencefig}{\begin{tikzpicture}
        \begin{axis}[
            width=2.5in,
            height=2.5in,
        axis lines = left,
        xlabel = {Value $\vali[j]$},
        ylabel = {Regret $\regreti[j]$},
        xmin=0, xmax=3,
        ymin=0, ymax=0.5,
        xtick={0},
        ytick={0},
        clip mode=individual
    ]

\addplot[name path=blue, blue, domain=0:3] {.2*x + .1};
    \addplot[name path=red, red, domain=0:3] {.1*x + .2};
    \addplot[name path=green, green, domain=0:3] {-.1*x + .3};
    \addplot[name path=top,domain=0:3] {4};

    \addplot[black, thick, domain=0:0.5] {-.1*x + .3};
    \addplot[black, thick, domain=0.5:1] {.1*x + .2};
    \addplot[black, thick, domain=1:3] {.2*x + .1};

    \addplot[gray!30] fill between[of=green and top, soft clip={domain=0:0.5}];
    \addplot[gray!30] fill between[of=red and top, soft clip={domain=0.5:1}];
    \addplot[gray!30] fill between[of=blue and top, soft clip={domain=1:3}];

    \addplot [mark=*,black] coordinates {(0.5,.25)};
    \draw (axis cs:0.5,0) -- (axis cs:0.5,-0.02) node[below] {$\evali[j]$};
    \draw (axis cs:0,0.25) -- (axis cs:-0.1,0.25) node[left] {$\eregreti[j]$};

    \end{axis}
\end{tikzpicture}}

\newcommand{\pricecompfig}{\begin{tikzpicture}
  \begin{axis}[
      width=2.5in,
      height=2.5in,
        axis lines = left,
        xlabel = {Seller 1},
        ylabel = {Seller 2},
        xmin=0, xmax=1,
        ymin=0, ymax=1,
        xtick={0.2,0.4,1},
        ytick={0.3,0.6,1},
        xticklabels={$c_1$,$p_1$,1},
        yticklabels={$c_2$,$p_2$,1},
    ]

    \addplot[fill=lightgray,draw=none] coordinates {(0.4,0) (0.4,0.6) (0.8,1) (1,1) (1,0)} \closedcycle;

\addplot[very thick, dashed, gray] coordinates {(0.4,0) (0.4,0.6)};
    \addplot[very thick, dashed, gray] coordinates {(0,0.6) (0.4,0.6)};
\addplot[very thick, dashed, gray] coordinates {(0.4,0.6) (0.8,1)};

\addplot[black] coordinates {(0.4,0) (0.4,1)};
    \addplot[black] coordinates {(0,0.6) (1,0.6)};

\node at (axis cs:0.7,0.35) {$A_1$};
    \node at (axis cs:0.2,0.8) {$A_2$};
    \node at (axis cs:0.2,0.35) {$\emptyset$};

      \addplot [mark=*,black] coordinates {(0.75,.75)} node[above right] {$\vals$};
    
      \draw [black,thick,arrows=<->] (axis cs: 0.75,.75) -- (axis cs: 0.75,0.6);
      \draw [black,thick,arrows=<->] (axis cs: 0.75,.75) -- (axis cs: 0.4,0.75);

    \end{axis}
\end{tikzpicture}}

\newcommand{\pricecompuniformfig}{\begin{tikzpicture}
    \begin{axis}[
        width=2.5in, height=2.5in,
        axis lines = left,
        xlabel = {Seller 1},
        ylabel = {Seller 2},
        xmin=0, xmax=1,
        ymin=0, ymax=1,
        xtick={0.1,0.5,0.6,1},
        ytick={0.2,0.55,0.66,1},
        xticklabels={$c_1$,$p_1^{eq}$,{$p_1^c$},1},
        yticklabels={$c_2$,$p_2^{eq}$,{$p_2^c$},1},
    ]

\addplot[black, thick] coordinates {(0.5,0) (0.5,0.55)};
    \addplot[black, thick] coordinates {(0,0.55) (0.5,0.55)};
    \addplot[black, thick] coordinates {(0.5,0.55) (0.95,1)};

\addplot[thick, dashed] coordinates {(0.6,0) (0.6,0.66)};
    \addplot[thick, dashed] coordinates {(0,0.66) (0.6,0.66)};
    \addplot[thick, dashed] coordinates {(0.6,0.66) (1,1.06)};

    \end{axis}
\end{tikzpicture}}

\newcommand{\vposteriorscore}[4][]{\pgfmathsetmacro{\psMraw}{#2}\pgfmathsetmacro{\psXraw}{#3}\pgfmathsetmacro{\psYraw}{(\psXraw < \psMraw) ? (1 - \psXraw/(2*\psMraw)) : (\psXraw/(2*\psMraw))}\edef\psX{\psXraw}\edef\psY{\psYraw}\addplot[dashed,#1] coordinates {(\psX,0) (\psX,\psY)};
  \addplot[only marks,mark=square*,mark size=2pt,#1] coordinates {(\psX,\psY)};
\if\relax\detokenize{#4}\relax
\else
    \addplot[dashed,#1] coordinates {(0,\psY) (\psX,\psY)};
    \addplot[draw=none,forget plot] coordinates {(0,\psY)} node[anchor=east,font=\small]{#4};
  \fi
}

\newcommand{\vmisreportscore}[4][]{\pgfmathsetmacro{\m}{#2}\pgfmathsetmacro{\r}{#3}\pgfmathsetmacro{\p}{#4}\pgfmathsetmacro{\useLeft}{(\r <= \m) ? 1 : 0}\pgfmathsetmacro{\yR}{\useLeft*(1 - \r/(2*\m)) + (1-\useLeft)*(\r/(2*\m))}\pgfmathsetmacro{\yP}{\useLeft*(1 - \p/(2*\m)) + (1-\useLeft)*(\p/(2*\m))}

\addplot[only marks,mark=o,mark size=3.5pt,thick,#1] coordinates {(\r,\yR)};
  \addplot[only marks,mark=*,mark size=2pt,#1]       coordinates {(\p,\yP)};

\addplot[dashed,#1] coordinates {(\r,0) (\r,\yR)};

\pgfmathsetmacro{\eps}{0.0005} \pgfmathtruncatemacro{\sameXY}{
    (abs(\p-\r) < \eps) && (abs(\yP-\yR) < \eps) ? 1 : 0
  }

  \ifnum\sameXY=0
\pgfmathsetmacro{\dx}{\p-\r}\pgfmathsetmacro{\dy}{\yP-\yR}\pgfmathsetmacro{\len}{sqrt(\dx*\dx+\dy*\dy)}\pgfmathsetmacro{\xm}{(\r+\p)/2}\pgfmathsetmacro{\ym}{(\yR+\yP)/2}\pgfmathsetmacro{\amp}{0.08*\len}\pgfmathsetmacro{\safeLen}{max(\len,0.001)}\pgfmathsetmacro{\ux}{-\dy/\safeLen}\pgfmathsetmacro{\uy}{ \dx/\safeLen}\pgfmathsetmacro{\xctrl}{\xm + \amp*\ux}\pgfmathsetmacro{\yctrl}{\ym + \amp*\uy}\edef\RR{\r}\edef\PP{\p}\edef\YR{\yR}\edef\YP{\yP}\edef\XC{\xctrl}\edef\YC{\yctrl}\addplot+[
      no marks, dashed, ->,shorten >=2pt, smooth,#1
    ] coordinates {(\RR,\YR) (\XC,\YC) (\PP,\YP)};
  \fi
}

\newcommand{\vmisreportpt}[4]{\pgfmathsetmacro{\vpM}{#2}\pgfmathsetmacro{\vpR}{#3}\pgfmathsetmacro{\vpQ}{#4}\pgfmathsetmacro{\vpUseLeft}{(\vpR <= \vpM) ? 1 : 0}\pgfmathsetmacro{\vpY}{\vpUseLeft*(1 - \vpQ/(2*\vpM)) + (1-\vpUseLeft)*(\vpQ/(2*\vpM))}\expandafter\xdef\csname #1@x\endcsname{\vpQ}\expandafter\xdef\csname #1@y\endcsname{\vpY}}

\newcommand{\vcombpts}[5][]{\pgfmathsetmacro{\xone}{\csname #2@x\endcsname}\pgfmathsetmacro{\yone}{\csname #2@y\endcsname}\pgfmathsetmacro{\xtwo}{\csname #3@x\endcsname}\pgfmathsetmacro{\ytwo}{\csname #3@y\endcsname}\pgfmathsetmacro{\pp}{#4}\addplot[very thick,dashed,#1] coordinates {(\xone,\yone) (\xtwo,\ytwo)};
\pgfmathsetmacro{\den}{\xtwo - \xone}\pgfmathsetmacro{\ypRaw}{
    (\den == 0) ? \yone : ( \yone + ((\pp - \xone)/\den) * (\ytwo - \yone) )
  }\edef\yp{\ypRaw}\addplot[dashed,#1] coordinates {(0,\yp) (1,\yp)};
  \addplot[draw=none,forget plot] coordinates {(0,\yp)}node[anchor=east,font=\small]{#5};
\addplot[only marks,mark=*,mark size=2pt,#1]        coordinates {(\xone,\yone)};
  \addplot[only marks,mark=*,mark size=2pt,#1]        coordinates {(\xtwo,\ytwo)};
  \addplot[only marks,mark=square*,mark size=2pt,#1]  coordinates {(\pp,\yp)};
}

\usepackage{caption}

\title{The Economics of No-Regret Learning Algorithms \thanks{Thanks
    to Michihiro Kandori and the Econometric Scociety for inviting
    this review article to the 2025 World Congress of the Econometric
    Society.  Thanks to the
    participants of the world congress for insightful discussion and questions.  Thanks to
    Nicole Immorlica for discussing the article and providing
    exceptional recommendations.  This review article is appearing in
    {\em Advances in Economics and Econometrics: Thirteenth World
      Congress, Volume 2}.}}

\author{Jason Hartline}
\date{}

\begin{document}

\maketitle

\begin{abstract}
A fundamental challenge for modern economics is to understand what
happens when actors in an economy are replaced with algorithms.  Like
rationality has enabled understanding of outcomes of classical
economic actors, no-regret can enable the understanding of outcomes of
algorithmic actors.  This review article covers the classical computer
science literature on no-regret algorithms to provide a foundation for
an overview of the latest economics research on no-regret algorithms,
focusing on the emerging topics of manipulation, statistical
inference, and algorithmic collusion.\\[2ex]
\noindent {\bf Keywords:} no-regret learning, multi-armed bandit, swap
regret, correlated equilibrium, econometrics, Stackelberg equilibrium, manipulation, algorithmic collusion.
\end{abstract}

\section{Introduction}
\label{s:intro}

Everything is an algorithm.  At a high level, an algorithm takes
inputs, processes them, and produces outputs.  For a decision maker,
inputs might be preferences, signals, and beliefs; while outputs might
be actions.  For an economy, inputs are actions of individuals and
realizations of state variables such as supply, while outputs are
outcomes according to the aggregation rule of the economy.  Firms and
individuals can be thought of as taking actions via algorithms.  For
example, the algorithm could be Bayes-Nash best response: Bayesian
update based on signals, calculate equilibrium strategies of others,
and best respond according to that equilibrium.  More so, individuals
and firms are more often now literally employing computerized
algorithms to make these decisions.  Perhaps in a repeated game, a
computerized algorithm will have inputs according to historical
actions and payoffs, and will repeatedly make the next decision.  For
example, sellers in an online marketplaces such as Airbnb, eBay,
Amazon, and Google Ads all employ algorithms.  Many of these
marketplaces even offer algorithms as products to their participants.
Of course the rules of the marketplaces themselves are also governed
by algorithms.

The field of algorithms within theoretical computer science aims to
make sense of the world through the perspective that everything is an
algorithm.  For economic scenarios such as the ones mentioned above,
this involves integrating central concepts in economics such as game
theory, mechanism design, and information design and asking how the
theory of algorithms gives new perspective and understanding.  A
classic contribution from the mid 2000s is the algorithmic complexity
of Nash equilibrium.\footnote{This work was recognized by the Game
Theory Society's 2008 Kalai Prize in Computer Science and Game Theory.
See further discussion in \Cref{s:repeated-games}.}  Under standard
assumptions of algorithmic complexity, in general two-player bimatrix
games, there is no algorithm that computes a Nash equilibrium in every
such game in a reasonable amount of time.  Thus, while Nash equilibria
are guaranteed to exist in all such games---a key property that has
led to their widespread adoption in economic theory---algorithms cannot
generally find them.  And if algorithms cannot find them, neither can
individuals, interacting or not.  After all, individuals and groups of
individuals are algorithms too (recall, they map inputs to outputs via
a process).

The premier venue for interdisciplinary research between economics
and computer science is the Association for Computing Machinery's
annual Symposium on Economics and Computation (EC).  Top papers from
computer science, operations research, and economics regularly appear
at EC before journal publication.  It is quite impossible in the
format of a single Econometric Society World Congress review article
to adequately summarize all the activities in this interface.  The
EC 2025 call for papers lists the following topics (in alphabetical
order): agent-based modeling, algorithmic fairness and data privacy,
{\bf auctions and pricing}, {\bf behavioral economics and bounded
  rationality}, blockchain and cryptocurrencies, contract design,
crowdsourcing and information elicitation, decision theory,
econometrics, {\bf economic aspects of learning algorithms}, economic
aspects of neural networks and large language models, {\bf equilibrium
  analysis and the price of anarchy}, {\bf equilibrium computation and
  complexity}, fair division, industrial organization, information
design, laboratory and field experiments, market design and matching
markets, market equilibria, mechanism design, {\bf online algorithms},
{\bf online platforms and applications}, social good and ethics,
social choice and voting theory, and social networks and social learning.
Instead, this review article will focus on one topic -- {\em the
  economics of no-regret learning algorithms} -- and be somewhat
comprehensive.  The areas above that include this topic are in bold.

Correlated equilibrium \citep{aum-74} is a tractable equilibrium notion and
corresponds to dynamics of simple ``no regret'' learning algorithms in
repeated games.  Correlated equilibrium is a canonical relaxation of
Nash equilibrium.  While Nash requires the distribution of actions
selected by players to be independent, i.e., the joint distribution of
play is an outer product, i.e., a rank-one matrix; correlated
equilibrium allows any joint distribution of play and requires that
each action of each player be in a conditional best response with
respect to this joint distribution.  In the same bimatrix games where
Nash equilibrium is intractable, correlated equilibria can be computed
easily by linear programs.  Moreover, \citet{FV-97} and \citet{HM-00}
show that there are simple learning algorithms that, via repeated
interaction, reach correlated equilibrium.  In fact, many standard so
called ``no regret'' algorithms from the extensive computer science
literature on online learning, i.e., learning from repeated
interactions---that are not explicitly designed for strategic
scenarios---result in the empirical distribution of play converging to
a coarse correlated equilibrium (a further relaxation of correlated
equilibrium, \citealp{MV-78}).  No regret---the condition that after
playing a game for a while, a player does not regret the actions they
took versus actions from benchmark strategies---is a relaxed notion of
rationality that is easy to achieve (via algorithms,
\citealp{han-57}).  It is in a sense a lower bound on rationality, as
rational best response also satisfies no regret.  Thus, it is well
accepted now, in the computer science literature, that the stage-game
equilibrium concepts of correlated and coarse correlated equilibrium
correspond to computationally and informationally tractable outcomes
of strategic interaction.

This article will review the basic definitions of correlated and
coarse correlated equilibrium, the classic algorithms that achieve
them in repeated interactions, and some new consequences for economics
that arise from these concepts.  Specifically, there are two canonical
formulations of regret, best-in-hindsight regret and swap regret.  For
these concepts, the review article highlights the following results:

\begin{itemize}

\item Swap regret algorithms result in correlated equilibrium
  \citep{FV-97,HM-00}, and best-in-hindsight regret algorithms result in
  coarse correlated equilibrium \citep[e.g.,][]{you-04}:

\item There is a reduction that converts any no best-in-hindsight
  regret algorithm into a no swap regret algorithm \citep{BM-07}.

\item These concepts differ in their robustness to manipulation in a
  Stackelberg sense.  First, any reasonable learning algorithm against
  a static Stackelberg leader will learn to best respond, giving the
  Stackelberg leader and follower (the algorithm) their Stackelberg
  equilibrium payoffs.  Second, no-best-in-hindsight-regret algorithms
  can be manipulated beyond the Stackelberg payoff (increasing the
  leader's payoff), while no-swap-regret algorithms guarantee no
  improvement on the leader's Stackelberg payoff
  \citep{BMSW-18,DSS-19}.

\item The relaxation of Nash equilibrium to no-regret or low-regret
  outcomes admits methods for structural inference that identify the
  rationalizable set of preferences and regrets that are consistent
  with empirical play \citep{NST-15}.

\item Algorithmic collusion \citep[e.g.,][]{CCDP-20} and how to
  regulate algorithms for collusion can be understood from no-regret
  guarantees and the non-manipulability of no-swap-regret algorithms
  \citep{HLZ-24,HWZ-25}.
\end{itemize}

This overview of no-regret learning algorithms presents the most
fundamental methods and results and connects them to the forefront of
current research at the interface between computer science and
economics.  The analysis provided herein focuses on developing the
main intuitions, with full proofs of the most fundamental results.

\section{Online Learning}
\label{s:online-learning}

Online learning is the problem of making sequential decisions in an
unknown and potentially adversarial environment, with the goal of
performing nearly as well as the best fixed action in hindsight.  The
classical online learning model is stateless, in that the actions only
affect the payoff in the current rounds and have no effect on payoffs
in future rounds.  In each of $n$ rounds, a decision maker chooses one
of $k$ possible actions, receives a payoff that depends on the chosen
action, and then observes the payoffs of all actions. The goal is for an algorithm's performance to approach that of the best action
in hindsight as the number of rounds grows.

The term {\em online} comes from the area of {\em online algorithms}
which refers to a family of algorithmic problems where the algorithm
must make decisions before all of the algorithm's inputs are revealed.

\subsection{Model}

Consider a learner repeatedly choosing one of $k$ actions $a \in A$ over $n$ rounds.
Let $\olv_\olact^i \in [0,h]$ denote the payoff of action $\olact$ in
round $i$. In each round $i$:
\begin{enumerate}
\item the learner chooses an action $\olact^i$,
\item observes all payoffs $\olvs^i = (\olv_1^i,\ldots,\olv_k^i)$, and
\item receives the payoff of the chosen action $\olv_{\olact^i}^i$.
\end{enumerate}
The total payoff of the algorithm is
\[
\ALG_n = \sum\nolimits_{i=1}^n \olv_{\olact^i}^i,
\]
while the total payoff of the best fixed action in hindsight is
\[
\OPT_n = \max_\olact \sum\nolimits_{i=1}^n \olv_\olact^i.
\]
The per-round \emph{regret} after $n$ rounds is defined as
\[
\Regret\nolimits_n
= \frac{1}{n}\!\left[\OPT_n - \ALG_n\right]
= \frac{1}{n}\!\left[\max\nolimits_{\olact}\sum\nolimits_{i=1}^n \olv_\olact^i
- \sum\nolimits_{i=1}^n \olv_{\olact^i}^i\right].
\]
An algorithm has \emph{vanishing regret} if on all input sequences;
$\olvs^1,\ldots,\olvs^n$; the regret satisfies $\lim_{n \to \infty}
\Regret_n = 0$.

This framework assumes a worst-case (adversarial) environment, which
will later be useful for reductions from more complex learning models.
The adversarial model contrasts with the i.i.d.\ stochastic model
where the payoffs in each round are drawn independently and
identically, but from an unknown distribution.  The adversarial
analysis of online learning is especially useful when attempting to
develop algorithms for more challenging learning problems.
Specifically, the next two sections will give algorithms for such
problems via a reduction to the online learning problem solved in this
section.  

\subsection{Follow the Leader, $\FTL$}

First, we see that the natural Follow the Leader ($\FTL$) algorithm
fails completely in the adversarial environment.

\begin{definition}[Follow the Leader, $\FTL$]
  Let $\cumval_\olact^i = \sum_{\roundalt=1}^i \olv_\olact^\roundalt$ be the cumulative payoff of action $\olact$ up to round $i$. In each round $i$, the algorithm chooses
$
\olact^i = \argmax_\olact \cumval_\olact^{i-1}.
$
\end{definition}

The follow the leader algorithm is, in fact, optimal in the
i.i.d.\ stochastic environment \citep{HJL-20}.  Our focus on the
adversarial model is motivated by applications that cannot be captured
by the i.i.d.\ model, and for these applications its failure is
significant.  The follow the leader algorithm is the one player
version of fictitious play.  It is also related to the stateless
Q-learning algorithm that has been recently studied in the economics
literature on algorithmic collusion.

The following example shows that the algorithm fails.  Moreover, as
established next, the pathology is not just with this algorithm but
with any deterministic algorithm.  Essentially, learning well against
an adversarial online input requires randomization.

\begin{example}\label{ex:ftl}
Consider the case of $k=2$ actions and payoffs, and assume without loss of generality that in the first round that $\FTL$ chooses action 1.

\begin{center}
\begin{tabular}{c|ccccccc}
Round   & 1 & 2 & 3 & 4 & 5 & 6 & \ldots\\
\hline
Action 1 & \bf 0 & 1 & \bf  0 & 1 & \bf  0 & 1 & \ldots\\
Action 2 & 1/2 & \bf 0 & 1 & \bf  0 & 1 & \bf  0 & \ldots
\end{tabular}
\end{center}
On this sequence $\OPT_n \approx n/2$ by choosing either action while
$\FTL_n = 0$ by choosing the actions in bold.  Thus, the per-round
regret of $\FTL$ is $\frac{1}{n}[\OPT_n - \FTL_n] \approx 1/2$, which
does not vanish with $n$.
\end{example}

\begin{theorem}\label{thm:deterministic-lb}
All deterministic online learning algorithms have constant per-round
regret in worst case.
\end{theorem}

\begin{proof}
  Since the algorithm is deterministic, in each round Nature can
  adversarially choose the payoff as $0$ for the action chosen by the
  algorithm and payoff $1$ to all other actions.  With this choice,
  clearly $\ALG_n = 0$ while $\OPT \geq n/2$ as there must be an
  action that is not chosen by the algorithm more than half the time.
  The per-round regret is at least $1/2$.
\end{proof}

This impossibility result implies that randomization is essential for
achieving vanishing regret in the adversarial model.

\subsection{Learning Algorithms}

A natural approach is to replace maximum with a soft-max:
multiplicatively increase (resp.\ decrease) the probability of
choosing actions with high (respectively low) past payoffs.  The
resulting algorithm is known as {\em exponential weights}, {\em
  multiplicative weights update}, and {\em Hedge} and originated in
\citet{LW-94} and \citet{FS-97}.

\begin{definition}[Exponential Weights, $\EW$]\label{def:EW}
Let $\epsilon>0$ denote the learning rate. In round $i$, define $\cumval_\olact^{i-1}=\sum_{\roundalt=1}^{i-1}\olv_\olact^\roundalt$ and choose action $\olact$ with probability
\[
\pr_\olact^i
= \frac{(1+\epsilon)^{\cumval_\olact^{i-1}/h}}{\sum_{j'} (1+\epsilon)^{\cumval_{j'}^{i-1}/h}}.
\]
\end{definition}

\begin{example}
  For $\epsilon = 1$ and $\olv_\olact^i \in \{0,1\}$, Exponential Weights
  doubles the weight of an action whenever it achieves payoff $1$.
  For the following payoffs (left), Exponential Weights calculates the
  weights (center) and action selection probabilities (right) as follows:

\begin{center}
\begin{tabular}{c|cccc}
 Round $i$ & 1 & 2 & 3 & 4 \\
\hline
$\olv_1^i$ & 1 & 1 & 0 & 0 \\
$\olv_2^i$ & 0 & 0 & 1 & 1
\end{tabular}
\quad
\begin{tabular}{c|cccc}
 Round $i$ & 1 & 2 & 3 & 4 \\
\hline
$2^{\cumval_1^{i-1}}$ & 1 & 2 & 4 & 4 \\
$2^{\cumval_2^{i-1}}$ & 1 & 1 & 1 & 2
\end{tabular}
\quad
\begin{tabular}{c|cccc}
 Round $i$ & 1 & 2 & 3 & 4 \\
\hline
$\pr_1^i$ & 1/2 & 2/3 & 4/5 & 2/3 \\
$\pr_2^i$ & 1/2 & 1/3 & 1/5 & 1/3 \\
\end{tabular}
\end{center}
The total expected payoff of Exponential Weights is
given by the dot product of all values and probabilities.  In this example:  $1/2 + 2/3 +
1/5 + 1/3 = 1.7$.
\end{example}

The following theorem and corollary show that Exponential Weights has vanishing best-in-hindsight regret.

\begin{theorem}\label{thm:ew-bound}
For payoffs in $[0,h]$, $k$ actions, and learning rate $\epsilon$:
\[
\expect{\EW} \geq (1-\epsilon)\OPT - \frac{h}{\epsilon}\ln k.
\]
\end{theorem}

Small $\epsilon$ slows adaptation to data, making the algorithm more
stable to adversarial noise but slower to learn good actions; large
$\epsilon$ yields rapid adjustment but higher sensitivity to
adversarial noise.  This tradeoff can be better understood by
comparing the performance of Exponential Weights on the extreme input
of \Cref{ex:ftl} and the extreme input given by the constant input
$\olvs^i = (1,0)$ for all rounds $i$.

\begin{corollary}\label{cor:ew-regret}
For $n$ rounds with payoffs in $[0,h]$, choosing $\epsilon = \sqrt{(\ln k)/n}$ yields
\[
\expect{\Regret(\EW)} \leq 2h\sqrt{\frac{\ln k}{n}}.
\]
\end{corollary}

\begin{proof}
Since $\OPT \leq hn$, \cref{thm:ew-bound} gives $\expect{\EW} \geq \OPT{} - \epsilon hn - \frac{h}{\epsilon}\ln k$. Balancing the error terms via $\epsilon hn = \tfrac{h}{\epsilon}\ln k$ yields $\epsilon=\sqrt{(\ln k)/n}$ and the bound $2h\sqrt{(\ln k)/n}$.
\end{proof}

From \Cref{cor:ew-regret}, it is evident that a slower learning rate
is beneficial when the number of rounds $n$ is large, whereas a faster
one is beneficial when the number of actions $k$ is large.  The proof
of \Cref{thm:ew-bound} follows from algebraic manipulations and
inequalities that, though simple, do not afford clear intuition.  The
remainder of the section will present a different algorithm that
affords more intuition.

\subsection{An upper-bound: Be the Leader}

A standard approach in developing approximately optimal algorithms is
to (a) identify a benchmark that upper bounds the optimal performance
and (b) follow the intuition of the benchmark to develop an algorithm
with provably similar performance. In the following definition, Be the
Leader ($\BTL$) is a benchmark, not an implementable algorithm.

\begin{definition}[Be the Leader, $\BTL$]
The be-the-leader benchmark chooses, at each round $i$, the action
that would be optimal \emph{after} observing round $i$: $\olact^i =
\argmax_\olact \cumval_\olact^i$.
\end{definition}

The subsequent example and theorem show that $\BTL$ upper bounds the
best-in-hindsight performance $\OPT$.  See
\Cref{fig:btl-vs-opt}.

\begin{example}\label{ex:btl-vs-opt}
For two actions with payoffs
\[
\begin{array}{c|cccc}
 & 1 & 2 & 3 & 4\\
\hline
\text{Action 1} & 0.4 & 0.3 & 0.3 & 1.0\\
\text{Action 2} & 0.2 & 0.2 & 1.0 & 0.2
\end{array}
\]
the cumulative performances after each round are $(0.4,0.7,1.7,2.7)$
for $\BTL$ and $(0.4,0.7,1.4,2.0)$ for $\OPT$.  $\BTL$ takes action 1 in rounds
1, 2, and 4 and obtains the full payoff from those rounds, denoted as
$\btl^i$ for round $i$.  $\OPT$ switches to track the leader in these
rounds, but only obtains the marginal increase in the cumulative
performance, denoted as $\opt^i$ in round~$i$.  
\end{example}

\begin{figure}
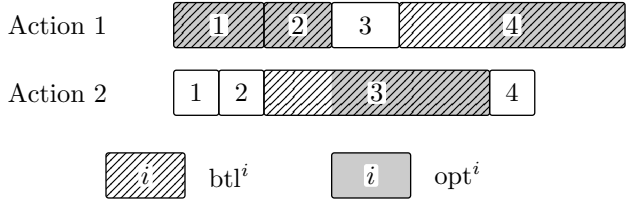

\centering
\btlvsoptfig
\caption{Illustration of \Cref{ex:btl-vs-opt} and \Cref{thm:btl-ge-opt}.  Payoff for each of the actions are depicted as the width of rectangles numbered by their round.  In each round $i$, the payoffs $\btl^i$ selected by $\BTL$ are striped, the payoffs $\opt^i$ accumulated by $\OPT$ are shaded gray, and we have $\btl^i \geq \opt^i$.}
\label{fig:btl-vs-opt}
\end{figure}

\begin{theorem}\label{thm:btl-ge-opt}
$\BTL \geq \OPT$.
\end{theorem}

\begin{proof}
Let $\OPT^i$ denote the best-in-hindsight value after $i$ rounds, and
$\opt^i = \OPT^i - \OPT^{i-1}$. I.e., $\opt^i$
equals the change in the leaders' cumulative payoffs before and after
round $i$.  Similarly define $\BTL^i$ and $\btl^i$ where $\btl^i$ is
the full payoff that the leader receives in round $i$. Thus, $\btl^i
\geq \opt^i$ for all $i$. Summing over all rounds gives $\BTL \geq
\OPT$.
\end{proof}

Intuitively, $\BTL$ is unimplementable in the online learning model
because it requires foresight, while $\FTL$ is implementable but can
perform poorly due to adversarial noise. To connect the two, we
introduce random perturbations to smooth out the noise.

\subsection{Perturbed Follow the Leader, $\FTPL$}

We now design an algorithm that approximates the unattainable $\BTL$ upper bound. The \emph{Perturbed Follow the Leader} ($\FTPL$) algorithm randomizes initial values to break ties and stabilize learning.  The analysis presented here is from \citet{KV-05}.

\begin{definition}[Perturbed Follow the Leader, $\FTPL$]\label{def:ftpl}
Fix a learning rate $\epsilon>0$. For each action $\olact$, hallucinate an initial payoff $\olv_\olact^0 = h \times Z_\olact$, where $Z_\olact$ is the number of consecutive tails obtained by flipping an $\epsilon$-biased coin (a geometric random variable). Let $\cumval_\olact^i = \sum_{\roundalt=1}^i \olv_\olact^\roundalt$. In each round $i$, choose
\[
\olact^i = \argmax_\olact \bigl(\olv_\olact^0 + \cumval_\olact^{i-1}\bigr).
\]
\end{definition}

\begin{example}
Revisiting \Cref{ex:ftl} with $\FTPL$:
\begin{center}
\begin{tabular}{c|c|ccccccc}
 & 0 & 1 & 2 & 3 & 4 & 5 & 6 & \ldots \\
\hline
\text{Action 1} & 2 & 0 & 1 & 0 & 1 & 0 & 1 &\ldots \\
\text{Action 2} & 3 & \bf 1/2 & \bf 0 & \bf 1 & \bf 0 & \bf 1 & \bf 0 &\ldots
\end{tabular}
\end{center}
The bold payoffs are selected by $\FTPL$ for the round 0 perturbation
shown.  On average over the perturbations, about half the time $\FTPL$
selects the bold payoffs and half the time it selects the other
payoffs; on average its payoff is about $n/2$.  Thus, $\FTPL \approx
n/2$.  Of course, $\OPT \approx n/2$ as well.
\end{example}

\begin{theorem}\label{thm:ftpl-main}
For payoffs in $[0,h]$,
\[
\expect{\FTPL} \geq (1-\epsilon)\OPT - \frac{h}{\epsilon}(1 + \ln k).
\]
\end{theorem}

\begin{corollary}\label{cor:ftpl-regret}
For $n$ rounds and payoffs in $[0,h]$, tuning $\epsilon = \sqrt{(1 + \ln k)/n}$ gives
\[
\expect{\Regret(\FTPL)} \leq 2h\sqrt{\frac{1 + \ln k}{n}}.
\]
\end{corollary}

\begin{proof}
The argument is the same as for \cref{cor:ew-regret}.
\end{proof}

Many distributions can be used to perturb the scores in
$\FTPL$;\footnote{In fact the exponential weights algorithm is
equivalent to a variant of $\FTPL$ with Gumbel perturbations.} the geometric
distribution is especially convenient because it is memoryless.
Specifically, for any integer $\ell$, a geometric random variable $X$
satisfies $\expect{X \given X \geq \ell} = \ell + \expect{X}$.

\begin{lemma}[Maximum of geometric variables,]\label{lem:geom-max}
Let $X_1,\ldots,X_k$ be independent geometric random variables with parameter $\epsilon$. Then
$\expect{\max_i X_i} \leq \frac{1 + \ln k}{\epsilon}$.
\end{lemma}

\begin{proof}
  If we flip coins of each of the $k$ geometric
  distributions in parallel rounds, then in each round we expect an
  $\epsilon$ fraction of the $k$ coins to succeed ({\em Heads}; and $1-\epsilon$
  fraction remain). The expected number of rounds until there are at most one
  left is at most $\log_{1/(1-\epsilon)} k$---by induction and
  Jensens' inequality---and the final one lasts $1/\epsilon$ rounds in
  expectation giving a total of $\frac{1 + \ln k}{\epsilon}$.
\end{proof}

\subsection{Analysis of $\FTPL$}

Our strategy for analysis of $\FTPL$ will be to introduce Perturbed Be
The Leader ($\BTPL$).  $\BTPL$ is the same as $\FTPL$, except that
round $i$ payoffs are included in the cumulative scores, i.e.,
$\olact^i = \argmax_\olact (\olv_\olact^0 + \cumval_\olact^{i})$.  We
will (a) relate $\FTPL$ and $\BTPL$ via a stability lemma and (b)
relate $\BTPL$ to $\OPT$ via a small perturbation lemma and
\Cref{thm:btl-ge-opt}.

\begin{lemma}[Stability]\label{lem:stability}
$\expect{\FTPL} \geq (1-\epsilon)\expect{\BTPL}.$
\end{lemma}

\begin{lemma}[Small perturbation]\label{lem:small-pert}
$\expect{\BTPL} \geq \OPT{} - \tfrac{h}{\epsilon} (1 + \ln k).$
\end{lemma}

\begin{figure}
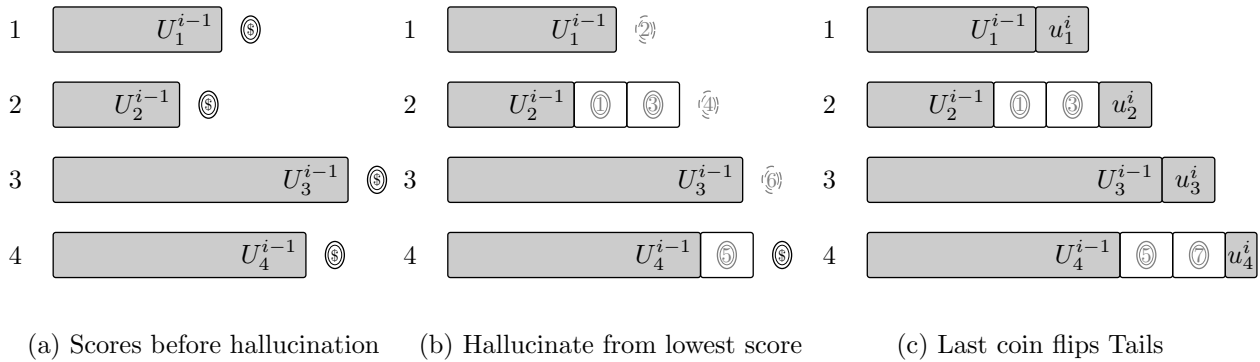

  \begin{subfigure}{0.31\textwidth}
    \pftlreadycoinsfig
    \caption{Scores before hallucination}
    \label{fig:pftl-ready}
  \end{subfigure}
  \begin{subfigure}{.33\textwidth}
    \pftlflippedcoinsfig
    \caption{Hallucinate from lowest score}
    \label{fig:pftl-flipped}
  \end{subfigure}
\begin{subfigure}{.33\textwidth}
\pftlfinalcoinfig
\caption{Last coin flips Tails}
\label{fig:pftl-final}
\end{subfigure}
\caption{The coupling argument of \Cref{lem:stability} depicted.  The
  coin flips depicted in the center figure are {\em Tails; Heads; Tails;
  Heads; Tails; Heads} (as numbered~1--6 in the figure).  Coins that
  flip {\em Heads} are dashed gray.  The coin flip in the right figure is
  {\em Tails} (numbered~7).  Depicted on the right, for these coin flips in
  round $i=4$, $\FTPL$ and $\BTPL$ obtain the same payoff.}
\end{figure}

\begin{proof}[Proof of \cref{lem:stability}]
Couple $\FTPL$ and $\BTPL$ under the same random perturbations and
break ties consistently. For each round $i$, begin with the
unperturbed scores and place a coin on each action (See
\Cref{fig:pftl-ready}), then add perturbations as follows: pick the
action with the lowest total score; flip its coin; if {\em Tails}
(probability $1-\epsilon$), add $h$ to its score, otherwise discard it
(on {\em Heads}); repeat until only one action $\olact^*$ remains (See
\Cref{fig:pftl-flipped}).  Flip $\olact^*$'s coin.  If it lands {\em
  Tails} (with probability $1-\epsilon$), then the best action’s score
exceeds the second-best by more than $h$, so both $\FTPL$ and $\BTPL$
choose the same action.  Hence, conditioned on the action of $\BTPL$,
$\FTPL$ agrees with $\BTPL$ with probability at least $1-\epsilon$;
and otherwise $\FTPL \geq 0$.  The lemma follows by linearity of
expectation.
\end{proof}

\begin{proof}[Proof of \Cref{lem:small-pert}]
  Let $\MAXPTRB = \max_\olact \olv_\olact^0$ denote the maximum hallucinated
  perturbation.  The same argument as for \cref{thm:btl-ge-opt} shows
  that $\BTPL + \MAXPTRB \geq \OPT$: consider the round 0 payoffs
  as real payoffs and then note that $\OPT$ without the hallucinations
  is less than $\OPT$ with the hallucinations.  Rearrange as $\BTPL
  \geq \OPT - \MAXPTRB$, take expectations, and apply
  \Cref{lem:geom-max} to get the desired bound.
\end{proof}

Combining \Cref{lem:stability,lem:small-pert} yields
\Cref{thm:ftpl-main}, establishing that $\FTPL$ achieves the same
asymptotic regret rate as Exponential Weights.  Stability
(\Cref{lem:stability}) gives robustness to adversarial noise and small
perturbation (\Cref{lem:small-pert}) means the algorithm learns to
track the best action quickly.

Note: the formal definition of Perturbed Follow the Leader (\Cref{def:ftpl})
randomizes only before making any decisions.  For this definition, the
vanishing-regret analysis above holds on any input that is fixed in
advance.  When employing this algorithm in games as discussed in later
sections, we will prefer a bound that holds against adaptive input
sequences.  It is easy to see from the analysis above that, on such
inputs the same bound can be proved for a version of the algorithm
that re-randomizes the hallucination in each round.

\section{(Online) Multi-Armed Bandit Learning}
\label{s:mab}

Multi-armed bandit learning extends online learning to the setting of \emph{partial feedback}. Whereas in full-feedback online learning the learner observes the payoffs of all actions after each round, in the bandit setting the learner only observes the payoff of the action played. This informational restriction creates a fundamental tradeoff between \emph{exploration} (learning about uncertain actions) and \emph{exploitation} (playing the best-known action).

\subsection{Model}

Consider $k$ actions $A$ and $n$ rounds. Let $\olv_\olact^i \in [0,h]$ denote the payoff of action $\olact$ in round $i$. In each round, the learner:
\begin{enumerate}
  \item chooses an action $\olact^i$,
  \item observes the realized payoff $\olv_{\olact^i}^i$, and
  \item receives that payoff.
\end{enumerate}
The total payoff is
$
\ALG_n = \sum_{i=1}^n \olv_{\olact^i}^i.
$

Identically to the online learning problem, the objective remains to
achieve vanishing regret against best fixed action in hindsight,
$\OPT_n = \max_\olact \sum_{i=1}^n \olv_\olact^i.$ The difference between the
online learning model of \Cref{s:online-learning} and the multi-armed
bandit model is that in the latter only $\olv_{\olact^i}^i$ is observed in
each round; the learner does not see the other $\olv_\olact^i$ for $\olact \neq
\olact^i$. Hence, if an action $\olact$ is not played, its quality cannot be
directly assessed.

It will be useful to have explicit notation for the probability
distribution over actions of the algorithm.  Let $\prs^i =
(\pr_1^i,\ldots,\pr_k^i)$ denote the vector of probabilities used to
sample an action in round $i$. The learner draws $\olact^i \sim \prs^i$.
On payoffs $\olvs^i = (\olv_1^i,\ldots,\olv_k^i)$ the learner obtains
expected payoff $\expect{\olv_{\olact^i}^i} = \sum_\olact \olv_\olact^i \pr_\olact^i =
\olvs^i \cdot \prs^i$ where $\olvs^i \cdot \prs^i$ is the vector dot
product.

\subsection{Reduction to Full-Feedback Learning}

We will \emph{reduce} the partial-feedback bandit setting to the
full-feedback online learning setting.  In this reduction the
multi-armed bandit algorithm (henceforth $\MAB$) will interact with an
online learning algorithm (henceforth $\OLA$).  The $\OLA$ will
recommend actions.  The multi-armed bandit algorithm will choose an
action that may or may not be what the $\OLA$ recommended.  The
multi-armed bandit algorithm will learn the payoff of the chosen
action.  It will then have to construct a full vector of payoffs, one
for each action, to feed into the $\OLA$.  This structure is depicted
in \Cref{fig:mab-reduction}.  The key difficulty of what to input into
the $\OLA$ for actions that are not played will be solved by producing
unbiased estimates of the payoffs of all actions.

\begin{figure}
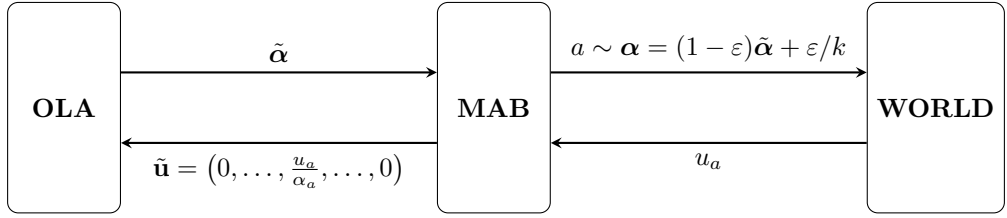

\centering
\mabreductionfig
\caption{Reduction from multi-armed bandit learning to full-feedback online learning.}
\label{fig:mab-reduction}
\end{figure}

\subsubsection{Propensity Scoring}

We must specify what payoff vector to report to the underlying
full-feedback algorithm. Let $\eolvs^i$ denote the reported
(estimated) payoff vector. The goal is to ensure that $\eolvs^i$ is an
unbiased estimator of the true payoff vector $\olvs^i$.  This
technique, known as \emph{propensity scoring}, scales the observed
payoff inversely by its sampling probability and records a score of zero
for unobserved payoffs.

\begin{definition}[Propensity Scores]
  For distribution $\prs^i$ over realized actions $\olact^i$, the {\em
    propensity score} is the vector of estimates $\eolvs^i$ of true
  payoffs $\olvs^i = (\olv_1^i,\ldots,\olv_k^i)$ defined as
\[
\eolv_\olact^i =
\begin{cases}
  \nicefrac{\olv_{\olact}^i}{\pr_{\olact}^i} & \text{if } \olact=\olact^i,\\
  0 & \text{otherwise.}
\end{cases}
\]
\end{definition}

\begin{lemma}\label{lem:unbiased}
  The propensity score $\eolvs^i$ is an {\em unbiased estimator} of the true payoff $\olvs^i$, i.e.\ $\expect{\eolvs^i} = \olvs^i$.
\end{lemma}

\begin{proof}
For each action $\olact$,
\[
\begin{aligned}
\expect{\eolv_\olact^i}
&= \expect{\eolv_\olact^i \mid \olact^i = \olact}\,\prob{\olact^i = \olact}
  + \expect{\eolv_\olact^i \mid \olact^i \neq \olact}\,\prob{\olact^i \neq \olact}\\
&= \tfrac{\olv_\olact^i}{\pr_\olact^i} \pr_\olact^i + 0\,(1-\pr_\olact^i)
= \olv_\olact^i.
\end{aligned}
\]
Applying this equality to all actions $\olact$ gives $\expect{\eolvs^i} = \olvs^i$.
\end{proof}

In the reduction following \Cref{fig:mab-reduction}, the propensity
scores $\eolvs^i$ will be used as input to the online learning
algorithm.  Intuitively, an online learning algorithm that obtains
good performance against the best-in-hindsight action is also
able to do so when it only has access to unbiased estimates. Indeed,
this result will be shown in the subsequent analysis.

\subsubsection{Bounding the Effective Payoff Range}

The performance of the online learning algorithm $\OLA$ depends on the range
of payoffs it receives. (Recall \Cref{thm:ew-bound}.)  Denote this
range $[1,\tilde{h}]$. Using the propensity scores, the non-zero inputs to the
online learning algorithm are $\eolv_\olact^i =
\nicefrac{\olv_\olact^i}{\pr_\olact^i}$.  Note that rare actions, with small
$\pr_\olact^i$, can produce large estimated payoffs.  The range of these
payoffs is bounded by $\tilde{h} = \max_{i,\olact}
\nicefrac{\olv_\olact^i}{\pr_\olact^i}$.  If the probabilities $\pr_\olact^i$ are not
bounded away from zero, $\tilde{h}$ may be arbitrarily large, and the
regret bounds of the underlying algorithm become vacuous.

Fortunately, our multi-armed bandit algorithm is free to pick the
probabilities and it can ensure they do not get too small.  The most
straightforward approach to ensuring these probabilities are not too
small is to require a minimal probability $\epsilon$ of exploring a
random one of the $k$ actions.  Such random exploration would give a lower
bound on the probability of taking each action $\olact$ of $\pr_\olact^i \ge
\nicefrac{\epsilon}{k}$.

\begin{lemma}\label{lem:h-tilde} For payoffs $\olv_\olact^i \in [0,h]$,
if $\pr_\olact^i \ge \nicefrac{\epsilon}{k}$ for all $i,\olact$, then $\eolv_\olact^i \le \tilde{h} = \nicefrac{kh}{\epsilon}$.
\end{lemma}

\begin{proof}
Since $\olv_\olact^i \le h$, we have
$\eolv_\olact^i = \nicefrac{\olv_\olact^i}{\pr_\olact^i} \le \nicefrac{h}{\epsilon/k} = \nicefrac{kh}{\epsilon}$.
\end{proof}

The tradeoff is clear: ensuring sufficient exploration (large
$\epsilon$) keeps $\tilde{h}$ small, while excessive exploration
reduces the payoff from exploitation.

\subsubsection{Reduction Algorithm}

We can now combine unbiased estimation with minimal exploration. The
following reduction converts any full-feedback online learning
algorithm ($\OLA$) into a multi-armed bandit algorithm ($\MAB$).
Applying this reduction to the Exponential Weights algorithm gives a
variant of the {\em Exp3} algorithm of \citet{ACFS-02}.  In the
reduction, the learning rate and the exploration probability are
equated at $\epsilon$.

\begin{definition}[Reduction from $\MAB$ to $\OLA$]\label{def:mab-reduction}
In each round $i$:
\begin{enumerate}
  \item Obtain probabilities $\eprs^i$ from $\OLA$.
  \item Define the actual sampling probabilities $\prs^i$ as
  \[
  \pr_\olact^i = (1-\epsilon)\,\epr_\olact^i + \nicefrac{\epsilon}{k}.
  \]
  \item Draw an action $\olact^i \sim \prs^i$ and receive payoff $\olv_{\olact^i}^i$.
  \item Report the estimated payoff vector $\eolvs^i$ to $\OLA$ as
  \[
  \eolv_\olact^i =
  \begin{cases}
  \nicefrac{\olv_\olact^i}{\pr_\olact^i} & \text{if } \olact=\olact^i,\\
  0 & \text{otherwise.}
  \end{cases}
  \]
\end{enumerate}
\end{definition}

\begin{theorem}\label{thm:mab-to-ola}
Suppose the online learning algorithm $\OLA$ on payoffs in $[0,\tilde{h}]$ guarantees:
\[
\expect{\OLA} \ge (1-\epsilon)\widetilde{\OPT} - \smash{\frac{\tilde{h}}{\epsilon}}\ln k.
\]
Then, for payoffs in $[0,h]$, the reduced multi-armed bandit algorithm satisfies
\[
\expect{\MAB} \ge (1-2\epsilon)\OPT - \smash{\frac{h\,k}{\epsilon^2}}\ln k.
\]
\end{theorem}

Intuitively, by \cref{lem:unbiased}, the estimated payoffs are
unbiased. Using \cref{lem:h-tilde}, payoffs are bounded by $\tilde{h}
= \nicefrac{kh}{\epsilon}$ which plugs into the second term of the
regret bound.  Since we only follow the $\OLA$ with probability
$1-\epsilon$ we lose an additional $(1-\epsilon)$ on the entire bound
and the $(1-\epsilon) \OPT$ we would get if we always followed $\OLA$
degrades to $(1-\epsilon)^2\OPT$ which is more convenient to lower
bound as $(1-2\epsilon)\OPT$.  The formal proof given next will ensure
our intuition about the unbiased estimates is correct.

We can solve for the optimal learning and exploring rate $\epsilon$
and plug in the Exponential Weights algorithm to obtain the following
corollary.  This combined algorithm is a variant of Exp3, the
Exponential-weights algorithm for Exploration and Exploitation.  The
regret bound of the corollary is improved to $2.63\, h \, \sqrt{\frac{k}{n}
  \ln k}$ in the direct analysis of \citet{ACFS-02}---which applies a
better bound on the variance of the unbiased estimates, sets the
learning and exploration rates differently, and achieves the optimal
statistical dependence on the number of rounds $n$.

\begin{corollary}\label{cor:mab-vanishing}
For payoffs in $[0,h]$, by setting $\epsilon = \sqrt[3]{\frac{k}{n}\ln
  k}$ the multi-armed bandit algorithm based on exponential weights
(Exp3) satisfies vanishing per-round regret $$\Regret_n \le
3h\sqrt[3]{\frac{k}{n} \ln k}.$$
\end{corollary}

We complete the section with the formal proof of \Cref{thm:mab-to-ola}.

\begin{proof}[Proof of \Cref{thm:mab-to-ola}]
  This analysis applies the guarantee of the online learning algorithm
  ($\OLA$) on its input, which is the estimated payoffs, to derive a
  guarantee for the multi-armed bandit algorithm ($\MAB$) on the true
  payoffs.  Let $\eregret = \nicefrac{\tilde{h}}{\epsilon}\ln k$ be
  the additive loss on the estimated payoffs and let $\olact^* = \argmax_\olact
  \sum_i \olv_\olact^i$ denote the best fixed action in hindsight on the
  true payoffs. By the assumption of the theorem statement, the
  full-feedback algorithm $\OLA$ ensures for any input
  $\eolvs^1,\ldots,\eolvs^n$ and an arbitrary action $\olact^*$ (while
  $\olact^*$ is the optimal action for the true payoffs, it is an arbitrary
  action for the estimated payoffs):
\begin{align*}
\OLA
  &= \sum\nolimits_i \eprs^i \cdot \eolvs^i \qquad\quad \geq  (1-\epsilon) \sum\nolimits_i \eolv_{\olact^*}^i - \eregret
\\
\intertext{Taking expectations of both sides, we have:}
\expect[\eprs,\eolvs]{\OLA}
  &= \sum\nolimits_i \expect[\eprs,\eolvs]{\eprs^i \cdot \eolvs^i} \geq  (1-\epsilon) \sum\nolimits_i \expect[\eolvs]{\eolv_{\olact^*}^i} - \eregret
\\
  & \qquad \qquad \shortparallel \qquad\qquad \qquad \qquad\qquad \shortparallel \\
  &\quad \sum\nolimits_i \expect[\eprs]{\eprs^i \cdot \olvs^i} \quad \geq (1-\epsilon) \sum\nolimits_i \olv_{\olact^*}^i - \eregret. \addtag \label{eq:mab1}
\end{align*}
The equalities between the two lines follow from $\eolvs^i$ being an unbiased estimator of $\olvs^i$, even when conditioned on $\eprs^i$, i.e., $\expect[\eolvs^i,\eprs^i]{\eolvs^i \given \eprs^i} = \olvs^i$.

The multi-armed bandit algorithm satisfies,
\begin{align*}
\expect{\MAB}
= \sum\nolimits_i \expect{\prs^i \!\cdot\! \olvs^i}
  &= (1-\epsilon)\sum\nolimits_i \expect{\eprs^i \!\cdot\! \olvs^i}
     + \frac{\epsilon}{k}\sum\nolimits_{i,\olact} \olv_\olact^i\\
&\ge (1-\epsilon)\sum\nolimits_i \expect{\eprs^i \!\cdot\! \olvs^i}. \addtag \label{eq:mab2}
\end{align*}

Combine inequalities \eqref{eq:mab1} and \eqref{eq:mab2} with $(1-\epsilon)^2 \geq (1-2\epsilon)$ and $\eregret = \frac{\tilde{h}}{\epsilon} \ln k = \frac{h k}{\epsilon^2} \ln k$:
\[
\expect{\MAB}
  \ge (1-2\epsilon)\OPT - \eregret
  = (1-2\epsilon)\OPT - \smash{\frac{h\,k}{\epsilon^2}}\ln k. \qedhere
\]
\end{proof}

Notice that the adversarial bound of $\OLA$ was crucial for this
reduction.  Even if the original payoffs of the multi-armed bandit
problem are i.i.d.\ across rounds, the unbiased estimates are not.
Thus, the guarantee of an online learning algorithm for the
i.i.d.\ model cannot be extended to the multi-armed bandit setting by
this proof.

\section{Swap Regret}
\label{s:swap-regret}

The goal in this section is to extend the analysis of (full feedback)
online learning algorithms from minimizing best-in-hindsight regret
(also called external regret) to minimizing \emph{swap regret} (also
called internal regret).  Swap regret considers the total improvement
from considering each action, looking at the rounds in which the
action is used, swapping it with the best action in hindsight for
these rounds.  As we will see in subsequent sections, vanishing swap
regret algorithms have important economics properties.  This section
will return to the full feedback model of \Cref{s:online-learning} and
follow the analysis of \citet{BM-07}; however, the reduction of this
section can also be combined with the reduction of \Cref{s:mab} to give a
vanishing swap regret algorithm for the partial feedback model.

Swap regret minimization affords the following intuition
\citep{CHJ-20}.  A learner using different actions in different rounds
reveals the timing of payoff-relevant information resulting in this
distinct choice of actions.  Vanishing swap regret requires that the
learner has fully optimized with respect to this revealed information.

\subsection{Model}

Following the model of \Cref{s:online-learning}, an online learning
algorithm chooses, in each round $i = 1,\ldots,n$, a mixed strategy
$\pveci \in \distof{k}$, where $\pij$ is the probability that the
algorithm picks action $\olact$ in round $i$ and $\distof{k}$ is the
set of distributions over actions $\{1,\ldots,k\}$. Let $\vveci \in
[0,h]^k$ denote the realized payoff vector in round $i$.

\begin{definition}[Best-in-hindsight Regret]
An algorithm has \emph{best-in-hindsight regret} at most $\regret$
over $n$ rounds if, for all deviation actions $\olact \in \{1,\ldots,k\}$,
\[
\ALG = \sum\nolimits_{i=1}^n \pveci \cdot \vveci \geq \sum\nolimits_{i=1}^n \olv^i_\olact - \regret n.
\]
\end{definition}

\begin{definition}[Swap Regret] 
  An algorithm has \emph{swap regret} at most $\regret$ over $n$ rounds if, for any function
  $f: \{1,\ldots, k\} \to \{1,\ldots, k\}$ representing a deviation
  rule,
\[
\ALG = \sum\nolimits_{i=1}^n \pveci \cdot \vveci \geq \sum\nolimits_{i=1}^n \sum\nolimits_{\olact=1}^k \pij\, \olv^i_{f(\olact)} - \regret n.
\]
\end{definition}

For a fixed input, swap regret compares to what would have happened
had the learner consistently replaced each of its actions $\olact$
with some alternative $f(\olact)$. Clearly, minimizing swap regret is
a stronger condition than minimizing best-in-hindsight regret.  Thus,
algorithms with vanishing swap regret also have vanishing
best-in-hindsight regret.

\subsection{Reduction from Swap Regret to Best-in-hindsight Regret}

Recall, an important feature of the worst-case online learning model
is that more complex regret minimization problems, such as swap regret
minimization, can be reduced to best-in-hindsight regret
minimization. Again, this reduction does not hold under stochastic (i.i.d.)
assumptions, highlighting the relevance of the worst-case model.
The reduction  proceeds by maintaining
one best-in-hindsight algorithm for each action and combining them
through a stationary mixing rule that ensures consistency between the
probability of delegating to an algorithm $\olact$ to choose an action
and the probability that action $\olact$ is chosen (see \Cref{fig:swap-reduction}).

\begin{figure}[t]
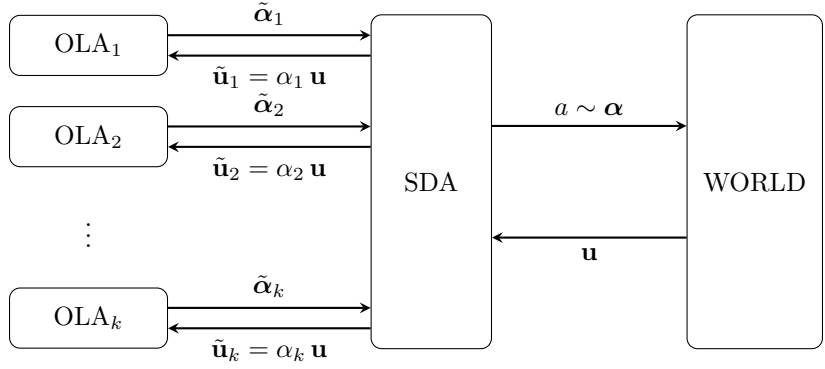

\centering
\swapreductionfig
\caption{Reduction from swap regret to best-in-hindsight regret.}
\label{fig:swap-reduction}
\end{figure}

A Markov chain on $k$ states is given by $k$ vectors of transition
probabilities $\eprs_1,\ldots,\eprs_k$ which can be organized in rows
of a row-stochastic matrix $\Qmat$.  Every finite Markov chain has at least
one stationary distribution $\pvec \in \distof{k}$ which can be viewed as a row
vector and satisfies $\pvec \Qmat = \pvec$.  In other words, if we
start with probabilities given by $\pvec$ and we take one transition
through $\Qmat$, the resulting distribution over states is again
$\pvec$.  Note, a stationary distribution $\pvec$ is also a principal
eigenvector of $\Qmat$, i.e., an eigenvector with eigenvalue equal to
1.  With such an $\pvec$, we have $(\prs \Qmat)_\olact = \sum_{\ell}
\pr_{\ell} \Qmat_{\ell \olact} = \pr_\olact$.

\subsection{Stationary Delegation Algorithm}
For learning with $k$ actions, the Stationary Delegation algorithm
($\SDA$) achieves low swap regret by delegating the decision to one of
a set of $k$ best-in-hindsight algorithms.  When delegated to,
algorithm $\OLA_{\olact}$ decides what action to take and
observes the payoffs of all actions (in expectation over the
delegation decision).  These algorithms are identical, but their
inputs are different, and thus, they will generally make distinct
recommendations.

\begin{definition}[Stationary Delegation Algorithm, $\SDA$]
Instantiate $k$ best-in-hindsight algorithms $\OLA_1, \ldots, \OLA_k$.  In round $i$:
\begin{enumerate}
\item Each algorithm $\OLA_\olact$ recommends a mixed strategy $\qvec^i_\olact \in \distof{k}$, and we collect these as the rows of the matrix:
\newcommand{\rowbar}{\raisebox{0.4ex}{\rule{1.2cm}{0.4pt}}}
\[
\Qmat^i =
\begin{bmatrix}
\rowbar\;\qvec^i_1\;\rowbar \\
\vdots                      \\
\rowbar\;\qvec^i_k\;\rowbar
\end{bmatrix}
\]
  
  \item Let $\pveci$ be the stationary distribution of $\Qmat^i$, satisfying $\pveci \Qmat^i  = \pveci$.
  \item The algorithm chooses an action by sampling $\olact^i \sim \pveci$. Equivalently, it first selects algorithm $\olact$ with probability $\pij$ and then selects an action according to $\qvec^i_\olact$.
  \item Input payoffs $\eolvs^i_\olact = \pij \vveci$ to each
    algorithm $\OLA_\olact$.
\end{enumerate}
\end{definition}

\begin{theorem}
\label{thm:sda-swap}
If each base algorithm $\OLA_\olact$ guarantees best-in-hindsight regret at most $\regret$, then the $\SDA$ algorithm achieves swap regret at most $k\regret$. That is, for all swap functions $f: \{1,\ldots, k\} \to \{1,\ldots, k\}$,
\[
\frac{1}{n} \sum_{i=1}^n \pveci \cdot \vveci \geq \frac{1}{n} \sum_{i=1}^n \sum_{\olact=1}^k \pij\, v^i_{f(\olact)} - k\regret.
\]
\end{theorem}

\begin{proof}
For any $\olact$ and $f(\olact)$, each base algorithm $\OLA_\olact$ satisfies
\begin{align*}
  \sum_{i=1}^n \pij (\qvec^i_\olact \cdot \vveci) &\geq \sum_{i=1}^n \pij v^i_{f(\olact)} - \regret n.
  \addtag \label{eq:swap}
  \\
\intertext{Summing the left-hand side over all $\olact$ gives}
\sum_{i=1}^n \sum_{\olact=1}^k \pij (\qvec^i_\olact \cdot \vveci)
&= \sum_{i=1}^n (\pveci \Qmat^i ) \cdot \vveci
= \sum_{i=1}^n \pveci \cdot \vveci,\\
\intertext{where the last equality holds because $\pveci$ is the stationary distribution of $\Qmat^i$. Thus, summing both sides of inequality \eqref{eq:swap} over $\olact$ yields}
\sum_{i=1}^n \pveci \cdot \vveci &\geq \sum_{i=1}^n \sum_{\olact=1}^k \pij v^i_{f(\olact)} - k\regret n,
\end{align*}
and proves the theorem.
\end{proof}

Combining \Cref{thm:sda-swap} with \Cref{thm:ew-bound} gives the
following corollary.  Though we will not give the details, it is
possible to improve this corollary to obtain a swap regret bound of
$2h \sqrt{\frac{k \ln k}{n}}$, i.e., moving the $k$ inside the
square-root.  This improvement follows from instantiating the
reduction with best-in-hindsight learning algorithms that obtain the
optimal regret bound without knowledge of the time horizon $n$.  In
fact, both Perturbed Follow the Leader and Exponential Weights can be
modified to satisfy this condition.

\begin{corollary}
\label{c:sda-swap}
The $\SDA$ reduction instantiated with Exponential Weights as the
online learning algorithm achieves per-round swap regret at most $2kh
\sqrt{\frac{\ln k}{n}}$.
\end{corollary}

\section{Learning in Repeated Games}
\label{s:repeated-games}

A repeated game is one where a stage game is repeated many times.  A
standard question of learning in games is whether there are natural
procedures by which players playing in a repeated game with initially
unknown payoffs will eventually learn to play an equilibrium of the
stage game.  This section will take a similar approach looking at
convergence of the joint distribution of the empirical distribution of
play to standard static equilibrium concepts.

\subsection{Model}

Consider a repeated finite two–player bimatrix game in which the row
player chooses an action $a \in A$, and the column player chooses an
action $b \in B$.  The corresponding payoffs are denoted $U_R(a,b)$
for the row player and $U_C(b,a)$ for the column player.  On each day $i$,
\begin{enumerate}
\item Row and Column choose actions $a^i$ and $b^i$, respectively, and
\item observe realized payoffs; for example, Row observes $U_R(a^i,b^i)$.
\end{enumerate}
Average payoffs are computed over the $n$ days; for instance, Row’s
average payoff is $\nicefrac{1}{n}\sum_{i=1}^n U_R(a^i,b^i)$.

Recall that a mixed Nash equilibrium imagines the players
independently randomize their actions, i.e., the joint distribution of
actions is a product distribution.  Learning in repeated games does
not, in general, converge quickly to independent randomization.

\subsection{Non-convergence to Nash Equilibrium}

A central result in the computational study of Nash equilibrium is
that computing Nash equilibrium is not generally tractable \citep{DGP-09}.  This
impossibility implies that simple learning dynamics are not generally
going to quickly converge, in any sense, to Nash equilibria.  If even
centralized algorithms cannot efficiently compute such an equilibrium,
neither can learning algorithms.

\begin{theorem}\label{thm:nash-intractable}
Computing a Nash equilibrium of a bimatrix game is, in the worst case
and under standard assumptions in computational complexity,
computationally intractable 
\end{theorem}

\subsection{Equilibrium of No-regret Learning}

Following the analysis of \citet{FV-97} and \citet{HM-00}, we will show
that outcomes of no-swap-regret learning and
no-best-in-hindsight-regret learning correspond to correlated
equilibria and coarse correlated equilibria, respectively.

\begin{definition}[Correlated Equilibrium, CE]\label{def:CE}
A joint distribution over actions $(a,b)$ is a \emph{correlated equilibrium} if, when a mediator draws $(a,b)$ according to this distribution and privately recommends $a$ to the row player and $b$ to the column player, each player finds it optimal to follow the recommendation. Formally, for the row player,
\[
a \in \argmax_{a'} \, \expect[(a,b)]{U_R(a'\!,b) \mid a},
\]
and analogously for the column player.
\end{definition}

\begin{definition}[Coarse Correlated Equilibrium, CCE]\label{def:CCE}
A joint distribution over actions $(a,b)$ is a \emph{coarse correlated equilibrium} if, when a mediator draws $(a,b)$ from this distribution and privately recommends $a$ to Row and $b$ to Column, each player weakly prefers following the recommendation to committing \emph{ex ante} to any fixed action that ignores the recommendation. Formally, for the row player and all actions $a \in A$,
\[
\expect[(a,b)]{U_R(a,b)} \;\ge\; \expect[(a,b)]{U_R(a'\!,b)}
\quad\text{for all } a',
\]
and analogously for the column player.
\end{definition}

Mixed Nash equilibrium equals CE intersected with independent play and
also equals CCE intersected with independent play. Thus, Nash
$\subseteq$ CE $\subseteq$ CCE.  It is also known that in two-player
zero-sum games Nash $=$ CE $=$ CCE.  In the following No-tie
Rock-Paper-Scissors, these concepts are all distinct.

\begin{example}[No-tie Rock–Paper–Scissors]\label{ex:RPS}
Consider the modified rock–paper–scissors game with payoffs shown below.

\[
\begin{array}{r|ccc}
 & \text{Rock} & \text{Paper} & \text{Scissors}\\
\hline
\textbf{Rock} & \mathbf{-6},-6 & \mathbf{-1},1 & \mathbf{1},-1\\
\textbf{Paper} & \mathbf{1},-1 & \mathbf{-6},-6 & \mathbf{-1},1\\
\textbf{Scissors} & \mathbf{-1},1 & \mathbf{1},-1 & \mathbf{-6},-6
\end{array}
\]

The unique Nash equilibrium is uniform mixing. There are additional
correlated equilibria that are not Nash equilibria; for example,
uniform mixing over $\{(R,P),(R,S),(P,R),\allowbreak(P,S),\allowbreak(S,R),(S,P)\}$ yields a
payoff of $0$ from following the mediator and a best payoff of $0$ on
recommendation $P$. There are also coarse correlated equilibria that are not correlated equilibria:
uniform mixing over $\{(S,P),(P,S)\}$ yields a payoff of $0$ from
following the mediator, while ex ante deviations to $R$, $P$, and $S$ yield
payoffs $0$, $-7/2$, and $-5/2$, respectively.
\end{example}

\begin{theorem}\label{thm:regret-to-ce}
Play has no best-in-hindsight regret if and only if the empirical distribution of play is a CCE, and play has no swap regret if and only if the empirical distribution of play is a CE.
\end{theorem}

\begin{proof}
Fix a sequence $((a^1,b^1),\ldots,(a^n,b^n))$. For the CCE statement, no-regret for Row means that, for every $a'$,
\[
\sum_{i=1}^n U_R(a^i,b^i)\;\ge\;\sum_{i=1}^n U_R(a',b^i).
\]
Consider a mediator who picks $i$ uniformly from $\{1,\ldots,n\}$ and
recommends $a^i$ to Row and $b^i$ to Column. Then no-regret for Row
and Column holds if and only if the mediator induces a CCE, since the
inequalities coincide. Analogously, CE follows for no swap
regret play.
\end{proof}

Importantly, it is not the case that the distribution of play that the
algorithms choose converges (a.k.a.\ last iterate convergence); it is
the whole empirical distribution of play that converges.
Specifically, play that cycles satisfies this definition.

\begin{example}
Consider two players in the No-tie Rock-Paper-Scissors game
(\Cref{ex:RPS}) repeatedly best-responding to each other starting with
$(R,S)$.  The progression of subsequent action pairs is $(R,P)$,
$(S,P)$, $(S,R)$, $(P,R)$, $(P, S)$, $(R,S)$ repeating forever.  The
empirical distribution of play converges to the correlated equilibrium
given in \Cref{ex:RPS} that mixes evenly over the six no-tie action
pairs.  The distribution over action played in each round, which is a
pointmass on the best response to the previous round, never converges.
\end{example}

\section{Manipulation}
\label{s:manipulation}

A central question in the study of no-regret learning in games is
understanding the ways in which it is similar to or different from
classical models of strategic behavior.  One important aspect is
manipulation, where the classical model is Stackelberg equilibrium: a
leader commits to a strategy and a follower best responds.  This
section follows the analysis of \citet{DSS-19} and shows that in
complete-information games, vanishing best-in-hindsight regret
learners can be further manipulated while vanishing swap-regret
learners cannot.

\subsection{Model}

Consider a finite two–player bimatrix game in which the \emph{leader} (the
row player) chooses an action $a \in A$, and the \emph{follower} (the
column player) chooses an action $b \in B$.  The corresponding payoffs
are denoted $U_L(a,b)$ for the leader and $U_F(b,a)$ for the follower.

\begin{definition}[Stackelberg equilibrium] When the leader commits to a (possibily randomized) action, and the follower best responds, a \emph{Stackelberg equilibrium} is given by an optimal such strategy for the leader, and the
  \emph{Stackelberg value ($\SV$)} is their expected payoff.
\end{definition}

For ease of exposition, we will assume that the follower's action is
a unique best response to the leader's Stackelberg strategy.  If this
is not the case, the leader can still force their preferred action of
the follower as long as it is undominated.  

\begin{proposition}
All no-regret algorithms can be led to play the Stackelberg outcome.
\end{proposition}

\begin{proof}
If the leader plays its Stackelberg strategy in every round—i.e.,
plays statically—then the follower faces a stationary environment.  A
no-regret learner in a stationary environment converges to a best
response, which is precisely the follower’s Stackelberg response.
Hence play converges to the Stackelberg outcome.
\end{proof}

\begin{definition}[Manipulability]
  A learning algorithm is \emph{manipulable} if there
  exists a game and a leader strategy such that the difference between
  the leader's payoff and the Stackelberg value is positive and
  non-vanishing.
\end{definition}

Informally, algorithms that only guarantee vanishing best-in-hindsight
regret tend to be manipulable, whereas those that guarantee vanishing
swap regret are not.  Although vanishing swap regret implies vanishing
best-in-hindsight regret, the converse is not true, so to obtain lower
bounds we focus on a weaker subclass of best-in-hindsight learners.

\begin{definition}[Mean-based learner]
  A learning algorithm is \emph{mean-based} if, in each round, actions
  that are $\gamma$ worse than the best-in-hindsight action's
  per-round payoff are played with probability at most $\gamma$.
\end{definition}

Canonical examples of mean-based algorithms include Exponential
Weights ($\EW$) and Perturbed Follow the Leader ($\FTPL$).

\subsection{Manipulability of Mean-based Learners}

\begin{theorem}\label{t:meanbased}
Mean-based no-regret algorithms are manipulable.
\end{theorem}

\begin{proof}
Consider the following bimatrix game, with the leader (row player, bold) and learner (column player). 

\begin{center}
  \begin{tabular}{r|ccc}
           & Left & Mid & Right\\ \hline
{\bf Up}   & {\bf 0}, $\epsilon$ & {\bf -2}, -1 & {\bf -2}, 0 \\
{\bf Down} & {\bf 0}, -1 & {\bf -2}, 1 & {\bf 2}, 0
  \end{tabular}
\end{center}

The leader’s manipulation strategy is to first play \emph{Up} for
$m=n/2$ rounds and then play \emph{Down} for $m=n/2$ rounds.  On such
a sequence a mean based learner behaves similarly to Be the Leader
($\BTL$).  For simplicity we give the analysis for $\BTL$.  The main
difference between $\BTL$ and a mean-based learner is that the latter
will gradually shift probability as a new action starts to compete
with the best action in hindsight, while the former shifts it
immediately when the new action overtakes the best action in
hindsight.

For the first $m=n/2$ rounds, the leader plays \emph{Up}.  During these rounds BTL plays \emph{Left} and their cumulative payoffs across ({\em Left, Mid, Right}) become approximately $(\epsilon m,\,-m,\,0)$.  For the remaining $m$ rounds, the leader switches to \emph{Down}.  Initially, BTL continues to play \emph{Left}; after $\epsilon m$ additional rounds, the cumulative payoffs shift to $(0,\,-m+\epsilon m,\,0)$, at which point BTL switches to \emph{Right}.  The play ends with cumulative payoffs roughly $(-m+\epsilon m,\,0,\,0)$.  The leader’s total payoff is then approximately
\[
0 \cdot (m + \epsilon m) \;+\; 2\cdot(m-\epsilon m) \;\approx\; 2m \;=\; n.
\]
On the other hand, the Stackelberg value for this game is $0$.
To see this, consider the analysis of the leader's payoff as a function of their mixed strategy (defined by the probability that they play \emph{Up}):
\begin{center}
\begin{tabular}{c|c|c}
$\prob{\text{play Up}}$ & Learner Action & Leader Payoff \\
\hline
$[0,\frac12)$ & Mid  & $-2$ \\
$(\frac12,\frac{1}{1+\epsilon})$ & Right & $2 - 4 \prob{\text{play Up}} \leq 0$ \\
$(\frac{1}{1+\epsilon},1]$ & Left & $0$
\end{tabular}
\end{center}
Thus, against a mean-based learner the leader achieves a per-round
payoff bounded away from the Stackelberg value, establishing
manipulability of mean-based no-regret learners.
\end{proof}

Notice that a no-swap-regret algorithm would not be manipulated in the
above game.  This is because a swap-regret learner keeps track of
payoffs conditioned on the action they took.  When it switches
actions, it immediately begins gathering fresh data for the new action
rather than relying on stale cumulative averages.  In the game above,
after the initial $m+m\epsilon$ rounds in which \emph{Left} is played
the cumulative payoffs are $(0,\,-m+\epsilon m,\,0)$.  A switch to
\emph{Right} starts freshly collecting payoff ascribed to the right
action of $(-1,1,0)$ across ({\em Left, Mid, Right}), making \emph{Mid}
attractive right away instead of after the $2m$-th round; the leader’s
manipulation therefore fails.

\subsection{Non-manipulability of No-Swap-Regret Learners}

\begin{theorem}\label{t:noswap}
No-swap-regret algorithms are not manipulable.
\end{theorem}

The analysis will make use of the following generalization of the
Stackelberg value.

\begin{definition}[Optimistic Stackelberg value, $\OSV(b,r)$]
For a follower action $b$ and regret tolerance $r$, $\OSV(b,r)$, is the maximum leader payoff over mixed strategies $\rowstrat \in \Delta(A)$ such that $b$ is within $r$ of a best response:
\[
\expect[a\sim\rowstrat]{U_F(b,a)} \ge \max_{b'\in B}\expect[a\sim\rowstrat]{U_F(b',a)} - r.
\]
If action follower action $b$ is not forceable with regret at most $r$ then $\OSV(b,r) = -\infty$.
\end{definition}

The optimistic Stackelberg value has several important properties
that are captured by the subsequent fact, lemma, and corollary.

\begin{fact}\label{f:sv=maxosv}
 The Stackelberg value is $\SV = \max_{b \in B} \OSV(b,0)$.
\end{fact}

\begin{lemma}\label{l:osvLP}
For each $b$, $\OSV(b,r)$ can be expressed as a finite linear program
in which $r$ appears as a single linear constraint; hence $\OSV(b,r)$
is piecewise linear and concave in $r$ (where not $-\infty$).
\end{lemma}

\begin{corollary}\label{c:osvlinear}
There exists a constant $C$ such that for all $b$ and $r$, $\OSV(b,r) \le \SV{} + C r$.
\end{corollary}

\begin{proof}
  By \Cref{f:sv=maxosv}, all actions $b \in B$ satisfy $\OSV(b,0) \leq
  \SV$.  Each curve $\OSV(b,r)$ is piecewise linear, so
  $\max(\SV,\OSV(b,r))$ is piecewise linear in $r$. Combining,
  $\OSV(b,r)$ is upper bounded by $\SV{} + r\,C$ for some finite
  constant $C$.
\end{proof}

\begin{proof}[Proof of \cref{t:noswap}]
Consider the sequence $((a^1,b^1),\ldots,(a^n,b^n))$.  Let $\colstrat
\in \Delta(B)$ be the empirical frequencies of follower actions, and
let $\regretcond_b$ be the follower’s per-round regret conditional on
playing $b$.  The total average regret is $r = \regretvec \cdot
\colstrat$.  Let $\rowstrat^b$ be the leader’s conditional
distribution over actions when the follower plays $b$.  The leader’s
average payoff equals:
\[
\sum\nolimits_b \colprob_b \, \expect[a \sim \rowstrat^b]{U_L(a,b)}
\;\le\; \sum\nolimits_b \colprob_b \, \OSV(b,\regretcond_b)
\;\le\; \sum\nolimits_b \colprob_b [\SV{} + C\,\regretcond_b]
\;\le\; \SV{} + C\,r,
\]
where the inequalities use the definition of $\OSV$, \cref{c:osvlinear}, and \cref{f:sv=maxosv}.  As regret $r$ vanishes, the leader’s payoff approaches the Stackelberg value, so the leader cannot secure a constant per-round advantage; thus a no-swap-regret learner is not manipulable.
\end{proof}

Thus, we see that the follower's vanishing swap-regret guarantees that
in finite games of complete information that a Stackelberg leader
cannot obtain more than the Stackelberg value.  The rate at which the vanishing swap regret brings the leader's value to the Stackelberg value can be very slow (see \Cref{ex:manipulation-rate}, below). Interestingly, for Bayesian games and extensive form games, stronger conditions than no-swap-regret are needed to guarantee non-manipulability \citep{ACMM+25}.

\begin{example}\label{ex:manipulation-rate}
  Consider the following game with Stackelberg value $\SV{} = 0$:
  \begin{center}
  \begin{tabular}{l|cc}
                 & \em Left & \em Right \\ \hline
    \em \bf Up   & $1,0$ & $0,\epsilon$ \\
    \em \bf Down & $0,0$ & $0,0$
  \end{tabular}
  \end{center}
    On the other hand for $r \leq \epsilon$ the leader can achieve
    $\OSV(\text{Left}, r) = r/\epsilon$, i.e.\ $C=1/\epsilon$, an
    arbitrarily large number, in \Cref{c:osvlinear}.
\end{example}
    
\subsection{Notes}

The manipulability of learning algorithms was introduced by
\citet{BMSW-18} who defined the class of mean-based learners and
showed that, in the Bayesian pricing game between a buyer running a
mean-based learning algorithm and the seller attempting to manipulate the
algorithm to maximize revenue, the seller can extract the full surplus,
which exceeds the Stackelberg value (in this case, the Stackelberg
value is also known as the monopoly revenue, \citealp{mye-81}).  The
analysis above for normal form games, i.e., the manipulability of
mean-based algorithm and the non-manipulability of algorithms with
vanishing swap-regret is from \citet{DSS-19}.  This theory of
manipulability can be generalized to Bayesian games; however, a more
sophisticated notion of swap regret is needed.  See \citet{MMSS-22},
\citet{RZ-24}, and \citet{ACMM+25}.

\section{Inference for Learning Bidders}
\label{s:inference}

There is lots of evidence that bids by persistent bidders in frequent
auctions, such as the online advertisement auctions popularized by the
Google Ad Words platform, are not in equilibrium of the stage game.  These bids exhibit
cyclic behaviors and generally look like they come from algorithms.
This section follows the method of \citet{NST-15} for inferring
bidders' private values in repeated auctions from the assumption that
bidders have low best-in-hindsight regret.

Recall that if all bidders are low-regret learners, then the joint
distribution of play approximates a coarse correlated equilibrium
(\Cref{thm:regret-to-ce}) and that coarse correlated equilibrium is a
generalization of Nash equilibrium.  Thus, the development of
econometric methods for low-regret learners is a relaxation of the
classical econometric approach.

\subsection{Model}

There are $m$ bidders indexed by $j$ with fixed values $\vals =
(\vali[1],\ldots,\vali[m])$. The mechanism is repeated for $n$ rounds,
indexed by $i=1,\ldots,n$. In round $i$,
\begin{enumerate}
  \item bidders submit bids
    $\bids^i$,
  \item the mechanism applies an allocation rule $\ballocs^i$
    and a payment rule $\bpays^i$, and
  \item each bidder $j$ receives allocation $\balloci[j]^i(\bids^i)$,
    pays $\bpayi[j]^i(\bids^i)$, and has linear utility
    $U_j(\bids^i) = \vali[j]\,\balloci[j]^i(\bids^i)-\bpayi[j]^i(\bids^i)$.
\end{enumerate}
We assume bidders' behaviors satisfy low best-in-hindsight regret over
the $n$ rounds.

For bidder $j$, the average realized utility from the observed play $\bids^1,\ldots,\bids^n$ is
\[
\frac 1 n \sum\nolimits_i \bigl[\vali[j]\,\balloci[j]^i(\bids^i)-\bpayi[j]^i(\bids^i)\bigr].
\]
Let $\zee$ denote any alternative bid. Bidder $j$’s observed behavior exhibits per-round regret at most $\regreti[j]$ against hindsight bid $\zee$ if
\[
\frac 1 n \sum\nolimits_i \bigl[\vali[j]\,\balloci[j]^i(\bids^i)-\bpayi[j]^i(\bids^i)\bigr]
\;\ge\;
\frac 1 n \sum\nolimits_i \bigl[\vali[j]\,\balloci[j]^i(\zee,\bidsmi[j]^i)-\bpayi[j]^i(\zee,\bidsmi[j]^i)\bigr]
\;-\;\regreti[j].
\]

\subsection{Rationalizable Set}

An analyst who possesses the bid data $\bids^1,\ldots,\bids^n$ and
knows the stage allocation and payment rules
$\ballocs^1,\ldots,\ballocs^n$ and $\bpays^1,\ldots,\bpays^n$ would like
to infer the values of the bidders.  Unfortunately, while the analyst
is willing to assume that the bidders have low regret, the analyst
does not generally know how low their regrets are.  This lack of
knowledge motivates considering both values and regrets as parameters
and identifying the set of parameters that is consistent with the
data.  The goal is, thus, to characterize the set of values and
regrets that are rationalizable given the sequence of bids.

\begin{definition}[Rationalizable Set, $R_j$]
Bidder $j$'s {\em rationalizable set} $R_j$ is
the set of all pairs $(\vali[j],\regreti[j])$ such that the bid
sequence $\bids^1,\ldots,\bids^n$ satisfies the above regret
inequality for all deviations $\zee$.
\end{definition}

For each deviation $\zee$, define the average changes in allocation and payment that bidder $j$ would experience relative to the observed play:
\begin{align*}
\Delta \balloci[j](\zee)
  &= \frac{1}{n}\sum_i\!\left[\balloci[j]^i(\bids^i)-\balloci[j]^i(\zee,\bidsmi[j]^i)\right],\\
\Delta \bpayi[j](\zee)
  &= \frac{1}{n}\sum_i\!\left[\bpayi[j]^i(\bids^i)-\bpayi[j]^i(\zee,\bidsmi[j]^i)\right].
\end{align*}
Note that $\Delta \balloci[j](\zee)$ and $\Delta \bpayi[j](\zee)$ are
constants in the data.  Rearranging the regret inequality yields, for every $\zee$,
\[
\vali[j]\;\Delta \balloci[j](\zee)\;-\;\Delta \bpayi[j](\zee)\;\ge\;-\regreti[j].
\]
Each deviation $\zee$ thus imposes a linear constraint on $(\vali[j],\regreti[j])$.  See \Cref{fig:inference}.

\begin{figure}[t]
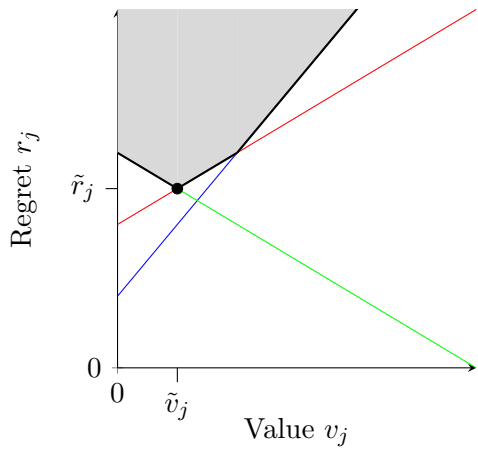

\centering
\inferencefig
\caption{Rationalizable set for bidder $j$ defined by linear inequalities.  The minimum rationalizable regret gives the value-regret pair $(\evali[j],\eregreti[j])$.}
\label{fig:inference}
\end{figure}

\begin{proposition}\label{prop:rationalizable-set}
The rationalizable set of bidder $j$,
\[
R_j \;=\; \bigl\{(\vali[j],\regreti[j]) \;:\; \vali[j]\Delta\balloci[j](\zee)-\Delta\bpayi[j](\zee)\ge -\regreti[j]\ \text{for all }\zee\bigr\},
\]
is convex.
\end{proposition}

\begin{proof}
Each inequality gives a half-space in $(\vali[j],\regreti[j])$; and an intersection of half-spaces is convex.
\end{proof}

\subsection{Identification and Interpretation}

Because the true $(\vali[j],\regreti[j])$ must lie in $R_j$, a reasonable way to identify one point is to select the value $\evali[j]$ that minimizes the implied regret $\eregreti[j]$:
\[
(\evali[j],\eregreti[j]) \;=\; \arg\min_{(\vali[j],\regreti[j])\in R_j}\ \regreti[j].
\]
This approach identifies the value consistent with the smallest possible regret.

If the learner has large regret, then a wider range of values $\vali[j]$ are consistent with observed play. If the learner has small regret, then fewer values are consistent. Extremely small regrets less than $\eregreti[j]$ (See \Cref{fig:inference}) are not rationalizable.

Note that in the full-feedback setting, i.e., when the counter factual
allocations and payments are known for every counterfactual bid
$\zee$, the rationalizable set is not a statistical object.  It holds
for the ex post realizations.  A similar method of inference can be
performed in the partial-feedback model, as long as the learners are
sufficiently exploring (See \Cref{s:mab}) and as long as the analyst
has access to the exploration probabilities.  In this case the bounds
are statistical because they are based on estimated counterfactual
allocations and payments.

\subsection{Notes}

The methods of \citet{NST-15}, described above, have been refined and
extended in several directions.  \citet{NN-17} give a method for point
identifying the values of the agents under the assumption that values
with lower regret are more likely.  They apply this method to several
data sets from human subjects experiments and observe that its
predictions outperform classical models.  \citet{ABMY-19} and
\citet{NS-21} apply these methods to bid simulation and prediction in
online ad auctions, respectively.  \citet{NS-21} show that their
method is better at predicting counterfactual behavior, e.g., when the
supply changes, than classical econometric
methods.

The methods for regulating algorithmic collusion from data that are
discussed subsequently in \Cref{s:collusion} generalize those of
\citet{NST-15} from full to partial feedback and from
best-in-hindsight to swap regret.  As small swap regret implies small
best-in-hindsight regret, the rationalizable set for small swap regret
bidders is contained in that for small best-in-hindsight regrets.
 
\section{Regulation of Algorithmic Collusion}
\label{s:collusion}

An important application area for online learning is pricing in
repeated markets.  While the families of algorithms developed in
\Cref{s:online-learning,s:mab,s:swap-regret} provably attain
equilibrium in the stage game (which defines competitiveness, see
\Cref{s:repeated-games}), empirical studies of other natural families of
algorithms have shown a tendency towards super competitive prices,
either by explicitly learning tit-for-tat or via statistical mistakes
(discussed more below in \Cref{s:collusion-notes}).  Given this
dichotomy of outcomes, a regulator may wish to enforce that
algorithms of the former family rather than the latter families are
employed.  This section presents the approach of \citet{HLZ-24} for
understanding and regulating algorithmic collusion.

Both the space of possible pricing algorithms and the space of
equilibria in repeated games are complex.  While this survey has
presented several online learning algorithms that could be used to
learn prices, there are certainly other algorithms that sellers may
use.  Moreover, it is well known that equilibria in repeated games
can sustain cooperation between players, even when non-cooperation is
a dominant strategy in the stage game.  For example, two sellers
could cooperate by both posting a high price even when posting a low
price dominates posting a high price (i.e., regardless of what the
other seller does, posting a low price is a best response).  It is
straightforward to construct algorithms that attempt to find and
sustain such cooperation.  Moreover, simple cooperative strategies can
be learned automatically without explicit algorithmic programming.

The current legal theory of collusion focuses on explicit agreements,
allowing other methods that might arrive at the same outcomes as tacit
collusion which are legal.  Nonetheless, a regulator may wish to limit
potential harm to consumers of supra-competitive prices by adopting
regulation.  Should such regulation be to restrict the space of
algorithms that are allowed?  If so, what are the allowed algorithms
and how should such restrictions be enforced?

\subsection{Model}

Consider the following simple model of imperfect price competition
with side information.  Sellers 1 and 2, with fixed marginal costs
$c_1$ and $c_2$, participate in a repeated market. In each round $i$,
both sellers receive private signals $\signal^i_1$ and $\signal^i_2$
that are potentially correlated with each other and the buyer's values
$\vals = (v^i_1,v^i_2)$ for each of the seller's products.  Each
seller $j$ then posts a price $p^i_j$. Buyer $i$ purchases from the
seller $j$ offering the higher net value,
\[
\argmax_j\, \{ v^i_j - p^i_j \},
\]
provided that the buyer’s utility is positive.  Hence, seller $i$
sells whenever the buyer's values $\vals = (v^i_1, v^i_2)$ lie in region $A_i$ of the
valuation space where the buyer prefers seller~$i$’s offer, as
depicted in \cref{f:pricecompfig}.

\begin{figure}[t]
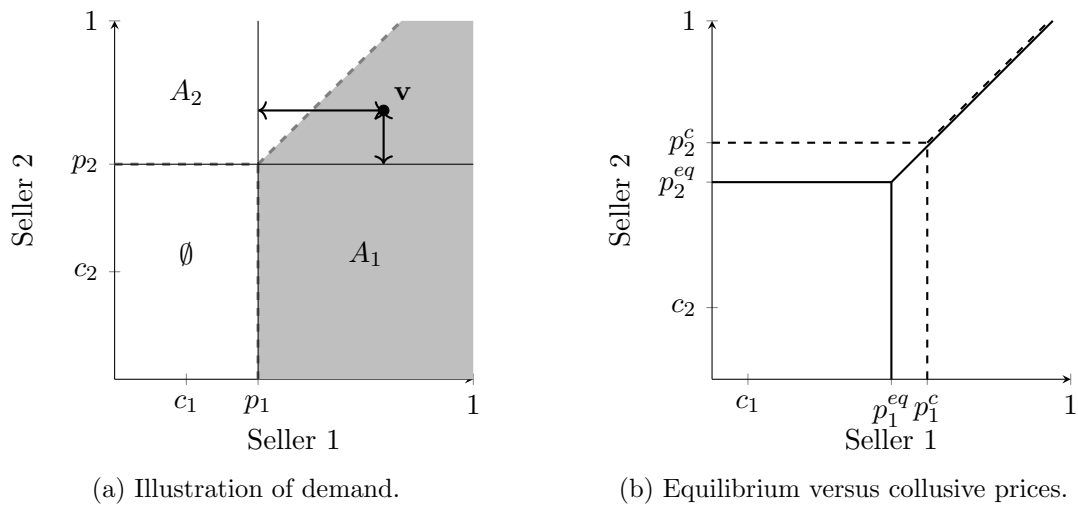

  \centering
  \begin{subfigure}{.45\textwidth}
    \centering
  \pricecompfig
  \caption{Illustration of demand.}
  \label{f:pricecompfig}
  \end{subfigure}
  \quad
  \begin{subfigure}{.45\textwidth}
    \centering
  \pricecompuniformfig
  \caption{Equilibrium versus collusive prices.}
  
  \label{f:pricecompuniformfig}
\end{subfigure}
\caption{On the left, the buyer purchases from seller~$i$ in the
  region $A_i \subseteq \mathbb{R}^2$ where $v_i - p_i \ge \max(v_{-i}
  - p_{-i},0)$.  Prices are shown with solid lines.  An example pair
  of values $\vals = (\val_1,\val_2)$ is depicted as a point, with distances to each
  price corresponding to the buyer's utilities.  The
  regions $\emptyset,A_1,A_2$ are separated by dashed lines; e.g.,
  region $A_1$ is shaded gray and contains the values $\vals$.  On the
  right, the example with costs $c_1 = 0.1$, $c_2 = 0.2$, and $\vals
  \sim U[0,1]^2$ is depicted along with equilibrium and collusive
  prices.  The collusive prices harm consumers because, relative to the equilibrium price, consumers are either
  excluded or are charged a higher price.}
\end{figure}

Both sellers can potentially benefit by setting prices above the
competitive level. For instance, consider an example with $c_1 = 0.1$
and $c_2 = 0.2$, where the buyers’ values are drawn from $\vals \sim
U[0,1]^2$. Figure~\ref{f:pricecompuniformfig} illustrates the Nash
equilibrium prices and the optimal supra-competitive prices attainable
through coordination.  The central policy concern is therefore whether
learning algorithms employed by the sellers converge to competitive or
non-competitive outcomes.

This model allows the possibility that sellers have private
information that is correlated with each other and the demand.  The
regulatory perspective of \citet{HLZ-24} is that the sellers might
possess such side information and that their conditioning their prices
on this side information is not collusion.  Thus, a model for
regulation of algorithmic collusion must allow it.  Of course, the
case that there is no side information is included in the model.

\subsection{Competitive Behavior}

There is opportunity for disagreement as to what sorts of behaviors
constitute collusion.  For example, there are markets where the only
Bertrand equilibrium is for the sellers to all post low prices; while
it is possible for a Stackelberg leader to commit to high prices, for
which followers best respond with medium-high prices, and all sellers
are better off.  This (Stackelberg) behavior and the resulting outcome
are not competitive, but is the behavior collusion?  Moreover, this
Stackelberg leader cannot be distinguished from an ignorant seller who
just posts a high price without understanding whether it is good or
bad.  On the other hand, (unilaterial) competitive behavior is
straightforward to define.  In the example above, both the Stackelberg
seller and the ignorant seller are non-competitive; while a
best-responding follower is being competitive.  It is straightforward
to argue that competitive behavior is not collusive.  Thus, we will
focus on regulating (unilateral) competitive behavior.

In games without side information, the canonical notion of competitive
outcomes is Nash equilibrium.  Unfortunately, Nash equilibrium is not
such a good notion for ensuring algorithms are competitive.  As we saw
in \Cref{s:repeated-games}, Nash equilibria are not generally
computationally tractable (\Cref{thm:nash-intractable}).  Thus, we
cannot generally expect algorithms, in a reasonable amount of time, to
reach a Nash equilibrium.

On the other hand, Nash equilibria is also too restrictive a concept
as it does not allow the sellers to have side information.
Acknowledging that sellers can possess side-information that is
potentially correlated with each other expands the set of possible
equilibria.  Specifically, the set of Bayes-Nash equilibria with all
possible information structures is the set of Bayes-correlated
equilibria \citep{aum-87}.

\begin{proposition}
  \label{prop:bne=bce}
  The set Bayes-Nash equilibria (BNE) of a stochastic game for all
  information structures equals the set of the Bayes-correlated
  equilibria (BCE).
\end{proposition}

\begin{proof}
  Any Bayes-Nash equilibrium with side-information induces a joint
  distribution on actions and the state.  Consider the mediator that
  recommends this joint distribution of actions conditioned on the
  state.  The Bayes-Nash property then implies the Bayes-correlated
  property.  Conversely, consider any mediator's conditional joint
  distribution on actions given the state that induces a BCE.
  Consider the side information that is equal to these actions, i.e.,
  each player's side information is their action in this BCE.  A
  possible joint strategy profile is for each seller to take the
  action that corresponds to their side-information.  This joint
  strategy is a BNE with the same outcome as the BCE.
\end{proof}

In \Cref{s:repeated-games}, we saw that (in games without information)
the set of correlated equilibria corresponds to the set of outcomes
where all players have vanishing swap regret
(\Cref{thm:regret-to-ce}).  The same result holds in games with
information for Bayes-correlated equilibria and players with vanishing
contextual swap regret.\footnote{The context is the signal that a
seller obtains.  Contextual swap regret allows the swap function to
depend on both the action and the context.  See \citet{BM-07}, this review
article will not discuss this model further.}  From the regulators
point of view, with the context hidden, vanishing contextual swap
regret looks like vanishing swap regret.  Thus, a unilateral
definition of competitive behavior that is achievable by algorithms is
vanishing swap regret.  Since our goal in identifying competitive
behavior is in regulating algorithmic collusion, it is important to
identify a unilateral property; it should be a choice that sellers can
make to not collude according to the definition.

\begin{definition}
  \label{d:competitive}
  A seller in a pricing game has {\em competitive behavior} if their actions satisfy vanishing swap regret.
\end{definition}

\subsection{Learning and Misspecification}

The literature in economics has focused on understanding algorithmic
collusion of a specific family of \emph{Q-learning} algorithms.
Q-learning is an online learning algorithm for Markov decision
processes.  Q-learning assumes that the environment has a known state
space but unknown, stationary transition dynamics and payoffs that are
independent and identically distributed conditional on the state.

\begin{definition}[Markov Decision Process, MDP]
  A Markov decision process is specified by a finite set of states, a
  set of actions, transition probabilities for each action-state pair
  to other states, and payoff distributions for each action-state
  pair.  When the decision maker is in a given state and taking a
  given action, the state transitions according to the transition
  probabilities and the decision maker obtains payoff according to the
  payoff distribution.  
\end{definition}

Note that in the definition of an MDP, the distribution of next state
and payoff is independent conditioned on the current state and action.
Assuming the feedback includes both the state and the payoff, and that
there exist policies that reach every state infinitely often, the
Q-learning algorithm will learn to maximize its payoff on the Markov
decision processes.

Empirical evidence suggests that some learning algorithms may converge
to supra-competitive prices. In one example, \citet{CCDP-20} show that
Q-learning with the \emph{state} defined as a one-round history can
learn a ``tit-for-tat'' pattern that sustains high prices.  This
result is important as it shows that algorithms that are not
explicitly programmed to sustain collusions can learn to do so.
Importantly, investigation of the source code of such algorithms will
not turn up any intent to collude.  Note that this application of
Q-learning is misspecified because the opponent's strategy is not
determined only by the state (one round of history).  In another
example, \citet{BS-22} consider stateless Q-learning and show that
even the statistical mistakes made because the model is misspecified
can produce supra-competitive prices.

Importantly, though these supra-competitive outcomes might be
equilibria of the repeated game (many outcomes are), they are not
equilibria of the stage game.  Thus, whether or not they are
considered collusive, they are definitely not competitive
(\Cref{d:competitive}).  In the remainder of this section we will see
that it is possible to design a regulation that can audit algorithms
for whether they are competitive.

\subsection{Designing Regulation}

The policy challenge is to regulate algorithmic pricing to ensure
competitiveness while allowing innovation and optimization. A
regulation should satisfy the following desiderata:
\begin{enumerate}
  \item it should permit the use of information and optimization;
  \item it should be unilateral, i.e., any seller can choose to
    be competitive (and non-collusive);
  \item it should be computationally reasonable, i.e., the computational requirements required to satisfy it are not onerous; and
  \item it should rely only on observable data rather than access to
    source code, preserving intelectual property and allowing
    algorithmic innovation.
\end{enumerate}

Recall that computing a Nash equilibrium is, in general,
computationally intractable (\Cref{thm:nash-intractable}). However, in
games with information, a Nash equilibrium coincides with a correlated
equilibrium (\Cref{prop:bne=bce}). Moreover, standard online learning
algorithms (\Cref{s:swap-regret}) are known to guarantee convergence
to correlated equilibria (\Cref{thm:regret-to-ce}).  This observation
suggests a promising regulatory approach: require that learning
algorithms used by market participants achieve vanishing \emph{swap
regret}. Recall that swap-regret-minimizing play converges to the set
of correlated equilibria.

A key detail is that costs and informational structures of sellers may
not be directly observable.  The legal theory of collusion suggests
that judges are reluctant to rule on what is in the minds of the
sellers.  Thus, a compatible theory for evaluating whether a seller is
being competitive should allow any consistent information structure or
costs, i.e., it should use revealed preference and revealed
information.  For revealed preference, the costs for which the swap
regret is the lowest can be inferred using the methods developed in
\Cref{s:inference}.  For revealed information, checking for vanishing
swap regret ensures that all information, that a seller reveals that
they have by selecting different prices, is optimally used.

\citet{HLZ-24} established that, given sufficient data, swap regret
can be estimated with known sample complexity bounds, enabling
empirical regulation of algorithmic behavior without inspecting
proprietary source code.  A key step in their statistical test is the method
of estimating the minimum rationalizable regret (\Cref{s:inference}
and \Cref{fig:inference}) by generalizing the methods of
\Cref{s:inference}.  The main takeaways from the theorem below which
measures the ``sample complexity'' of the statistical test is the
order of the dependence on various model parameters.  Some of the
strong assumptions of this theorem, such as the known discrete price
levels and minimum exploration probability, are relaxed by
\citet{HWZ-25}.

\begin{theorem}
  Given $k$ price levels, maximum price $\overline{\pay}$, minimum
  probability of exploring any price of $\underline{\pr}$, target regret
  $\overline{\regret}$, and failure probability $\delta$; there is a
  statistical test that with probability $1-\delta$ passes algorithms with regret at most
  $\overline{\regret}$ and fails algorithms with regret more than
  $2\overline{\regret}$ in number of rounds $n = O((\frac{k\,
    \overline{\pay}}{\underline{\pr}\,\overline{\regret}})^2(\log
  k/\delta))$.
\end{theorem}

It is also possible to regulate to force coarse correlated equilibria
by auditing for vanishing best-in-hindsight regret.  This has two
drawbacks, first it is a larger set than is necessary
(\Cref{prop:bne=bce}) so it admits potentially worse outcomes for
consumers.  Moreover, by the manipulation results of
\Cref{s:manipulation}, these worse-for-consumer outcomes can be found
by a Stackelberg leader manipulating the best-in-hindsight learner \citep{HWZ-25}.

In summary, the regulation of algorithmic collusion can be grounded in
the theory of no-regret learning. By mandating that learning
algorithms exhibit vanishing swap regret, regulators can ensure that
market outcomes correspond to competitive (stage-game) equilibria,
while maintaining flexibility and innovation in algorithm design. This
approach offers a data-driven, computationally tractable, and
information-based foundation for the governance of algorithmic
markets.

\subsection{Notes}
\label{s:collusion-notes}

See \citet{har-18} for a discussion of the legal theory around
algorithmic collusion as well as discussion of various methods of
regulating algorithmic collusion.  There is now an extensive
mostly-experimental literature in economics that studies various
families of algorithms and observes supra-competitive prices in some
configurations \citep{cal-27}.

One family of papers shows that algorithms that learn over internal
strategies that include collusive ones like tit-for-tat and
price-war-avoidance do learn to employ those strategies, even though
they are not explicitly programmed.  \citet{CCDP-20} does so with
Q-learning with one time-period history. \citet{CCDP-21} extends the
prior analysis to noisy observation. \citet{kle-21} extends to
asynchronous updates.  Finally, \citet{FGS-24} obtains this behavior
from pricing algorithms developed from large language models.

Another family of papers shows that algorithms with misspecified
models make statistical mistakes that result in supra-competitive
prices.  The statistical mistake that is made is assuming that payoffs
are i.i.d.\ when they are not, i.e., misspecification.  The main
result of \citet{AFP-24} compares stateless Q-learning with full
feedback and partial feedback where the partial feedback model makes
statistical mistakes due to misspecification; in their stylized price
competition setting, they find that full feedback gives competitive
prices while misspecified partial feedback gives supra-competitive
prices.  Thus, the statistical mistakes of Q-learning come from the
partial feedback model.\footnote{Note, the success of Q-learning with
full feedback, which is similar to the Follow the Leader algorithm of
\Cref{s:online-learning}, in simulations of stylized models is not
incompatible with its lack of theoretical guarantees in adversarial
models.  E.g., see \citet{BDO-24} a theoretical analysis of these
algorithms in Bertrand competition.}  Note that, without including the
last round of play as in \citet{CCDP-20}, collusive strategies
like tit-for-tat cannot be learned.  In parallel, \citet{BS-22} give
an empirical analysis of stateless Q-learning and \citet{BM-23} give
a theoretical analysis.  Other algorithms have been considered as
well.  \citet{HMP-21} gives an empirical analysis of the Upper
Confidence Bound (UCB) algorithm, which is a multi-armed bandit
algorithm with provable guarantees when the demand is i.i.d.\ across
rounds, but which is misspecified in pricing competition.

In contrast, \citet{HLZ-24} suggest that vanishing swap regret is a
unilateral definition of competitive behavior that, as we have seen in
\Cref{s:swap-regret}, algorithms can easily satisfy.  As discussed
above, the swap regret property can be audited from data.
\citet*{HWZ-25} extend the method to remove some technical assumptions
of \citet*{HLZ-24}.  As a caveat for this method, \citet{HWZ-25}
demonstrate, via a simulation of Q-learning in a setup similar to
\citet{BS-22}, that the audit can have false negatives when the true
costs of the sellers are unknown and the market converges to prices
that look like they are competitive for higher costs.  Nonetheless,
this definition of non-collusion enables regulators to identify
sellers that are running bad algorithms (that make statistical
mistakes) or ones that learn collusive strategies like tit-for-tat.

A central idea in \citet{HLZ-24,HWZ-25}---that information is revealed
in the actions of a no-regret agent and that vanishing swap regret
corresponds to optimal use of this revealed information---is from
\citet{CHJ-20}.  They relax the common prior assumption in
principal-agent mechanism design to no-regret learning by both the
principal and the agent.  \citet{CRS-24} improve on the computational
and statistical efficiency of the approach.

\bibliographystyle{apalike}

\end{document}